%% file: ArXiv.tex
\documentclass[11pt]{article}
\pdfoutput=1
\usepackage[margin=1in]{geometry}
\usepackage[onehalfspacing]{setspace}
\usepackage{amsmath}

\usepackage{hyperref}
\usepackage[capitalize]{cleveref}

\usepackage{graphicx,amssymb}
\usepackage{amsthm}
\usepackage{float}
\usepackage{xfrac}
\usepackage{xspace}
\usepackage{paralist}
\usepackage{enumerate,multicol,multirow}
\usepackage{cases}
\usepackage{caption,subcaption}

\usepackage{tikz,pgfmath}
\usepackage{times}
\usepackage{xurl}

\usetikzlibrary{matrix,positioning,quotes}
\usetikzlibrary{shapes.multipart}
\usetikzlibrary{calc}
\usetikzlibrary{arrows,decorations.markings}
\usetikzlibrary{matrix}
\usepackage[ruled,linesnumbered]{algorithm2e}

\usepackage[noend]{algpseudocode}
\usepackage{xcolor}

\usepackage{natbib}

\input{macro}

\title{Algorithmic Information Disclosure in Optimal Auctions\thanks{Yang Cai is supported by a Sloan Foundation Research Fellowship and NSF Awards CCF-1942583 (CAREER) and CCF-2342642. Yingkai Li is supported by a Sloan Foundation Research Fellowship. Jinzhao Wu is supported by a Research Fellowship from the Center for Algorithms, Data, and Market Design at Yale (CADMY) and the NSF Award CCF-2342642.}}

\author{Yang Cai\thanks{Department of Computer Science, Yale University. Email: \texttt{yang.cai@yale.edu}.}\\
\and Yingkai Li\thanks{Cowles Foundation for Research in Economics and Department of Computer Science, Yale University.
Email: \texttt{yingkai.li@yale.edu}.}
\and Jinzhao Wu\thanks{Department of Computer Science, Yale University.
Email: \texttt{jinzhao.wu@yale.edu}.}}

\date{}
\begin{document}

\maketitle
\begin{abstract}
This paper studies a joint design problem where a seller can design both the signal structures for the agents to learn their values, and the allocation and payment rules for selling the item. In his seminal work, \citet{myerson1981optimal} shows how to design the optimal auction with exogenous signals. We show that the problem becomes \emph{NP-hard} when the seller also has the ability to design the signal structures. 
Our main result is a \emph{polynomial-time approximation scheme (PTAS) for computing the optimal joint design} 
with at most an $\epsilon$ multiplicative loss in expected revenue. Moreover, we show that in our joint design problem, the seller can significantly reduce the information rent of the agents by providing partial information, which ensures a revenue that is at least $1 - \sfrac{1}{e}$ of the optimal welfare for all valuation distributions. 
\end{abstract}
\input{content/intro}
\input{content/prelim}
\input{content/nph}

\input{content/algorithms}
\input{content/conclusion}

\bibliographystyle{apalike}
\bibliography{ref}

\newpage
\appendix
\input{appendix/nph}
\input{appendix/algorithms}
\end{document}

%% file: macro.tex
\DeclareMathOperator*{\argmax}{argmax}
\DeclareMathOperator{\poly}{poly}

\newtheorem{definition}{Definition}
\newtheorem{corollary}{Corollary}[section]
\newtheorem{theorem}{Theorem}
\newtheorem{lemma}{Lemma}[section]

\newcommand{\notshow}[1]{{}}

\newcommand{\Xcomment}[1]{{}}

\newcommand{\half}{\frac{1}{2}}

\usepackage{color-edits}

\newcommand{\given}{|}

\newcommand{\prob}[2][]{\text{\bf Pr}\ifthenelse{\not\equal{}{#1}}{_{#1}}{}\!\left[{\def\givenn{\middle|}#2}\right]}
\newcommand{\expect}[2][]{\text{\bf E}\ifthenelse{\not\equal{}{#1}}{_{#1}}{}\!\left[{#2}\right]}

\newcommand{\tparen}{\big}
\newcommand{\tprob}[2][]{\text{\bf Pr}\ifthenelse{\not\equal{}{#1}}{_{#1}}{}\tparen[{\def\given{\tparen|}#2}\tparen]}
\newcommand{\texpect}[2][]{\text{\bf E}\ifthenelse{\not\equal{}{#1}}{_{#1}}{}\tparen[{\def\given{\tparen|}#2}\tparen]}

\newcommand{\sprob}[2][]{\text{\bf Pr}\ifthenelse{\not\equal{}{#1}}{_{#1}}{}[#2]}
\newcommand{\sexpect}[2][]{\text{\bf E}\ifthenelse{\not\equal{}{#1}}{_{#1}}{}[#2]}

\newcommand{\lbr}[1]{\left\{#1\right\}}
\newcommand{\rbr}[1]{\left(#1\right)}
\newcommand{\cbr}[1]{\left[#1\right]}

\newcommand{\abs}[1]{\left|#1\right|}
\newcommand{\indicate}[1]{\mathbf{1}\left[#1\right]}
\newcommand{\reals}{\mathbb{R}}

\newcommand{\naturals}{\mathbb{N}}

\renewcommand{\epsilon}{\varepsilon}

\newcommand{\val}{v}
\newcommand{\valspace}{V}
\newcommand{\dist}{F}
\newcommand{\density}{f}
\newcommand{\util}{u}
\newcommand{\alloc}{x}
\newcommand{\pay}{p}
\newcommand{\experiment}{\sigma}
\newcommand{\signal}{s}
\newcommand{\signals}{S}
\newcommand{\sets}{\mathcal{I}}
\newcommand{\dpf}{{\mathsf{Fea}}}
\newcommand{\dps}{{\mathsf{Imp}}}
\newcommand{\mech}{M}
\newcommand{\mechs}{\mathcal{M}}
\newcommand{\rev}{{\rm Rev}}

\newcommand{\optrev}{{\rm OPT\text{-} Rev}}
\newcommand{\optwel}{{\rm OPT\text{-}Wel}}
\newcommand{\wel}{{\rm Wel}}
\newcommand{\indic}{\mathbf{1}}

\newcommand{\probkname}{\textsc{Optimal $k$-Signal} }
\newcommand{\probname}{\textsc{Optimal Signal} }
\newcommand{\compensate}{c}
\newcommand{\grid}{\mathcal{G}}
\newcommand{\virtual}{\phi}
\newcommand{\virtualdist}{G}
\newcommand{\virtualdensity}{g}
\newcommand{\feasibles}{\mathcal{F}}
\newcommand\tl[2]{\genfrac{}{}{0pt}{}{#1}{#2}}
\newcommand{\discretize}{\textsc{Dis}}

\newcommand{\virtualdists}{\mathbf{\virtualdist}}
\newcommand{\vals}{\mathbf{\val}}
\newcommand{\compensates}{\mathbf{\compensate}}
\newcommand{\dists}{\mathbf{\dist}}
\newcommand{\experiments}{\mathbf{\experiment}}
\newcommand{\virtualwel}{\textrm{VW}}
\newcommand{\per}{\mathcal{P}}
\newcommand{\rd}{\textsc{RoundDown}}

%% file: content/intro.tex
\section{Introduction}
\label{sec:intro}

The classical auction theory primarily focuses on models that feature exogenous signal structures, wherein all participants in the auction are privately informed of their valuations for the items. Within this framework, \citet{myerson1981optimal} provides an elegant characterization of the optimal single-item auction, which also gives rise to an efficient algorithm for computing the optimal auction given the exogenous signal structure as input. However, in many practical applications, buyers initially lack information about their values for the items for sale due to insufficient data regarding item quality or suitability according to their preferences. 
Instead, these buyers depend on seller advertisements to gain insight into their own values. The nature of these advertisements varies depending on the context. For example, streaming service providers might offer free trials, thereby influencing consumers' perceived value of the service prior to committing to a subscription. Real estate agencies offer inspection opportunities to potential homebuyers, enabling better estimations of their value for a particular property. Similarly,  in the context of government auctions for oil field operations,  oil companies are often allowed to inspect and evaluate the potential oil reserves within a field prior to bidding.

In these applications, sellers can jointly design their advertising strategies and the subsequent auction mechanisms. In this paper, we focus on a canonical model introduced by \citet{bergemann2007information} for this joint design problem. In this model, a seller possesses a single item available for purchase, and there are $n$ heterogeneous buyers interested in it. 
These buyers have values independently drawn from a prior distribution. While the prior distribution is initially common knowledge between the seller and the buyers, the actual values of the buyers remain unknown to both parties. Instead, the seller can design Blackwell experiments, which enable the buyers to privately learn their own values. These Blackwell experiments, or equivalently signal structures, are mappings from the true values to a distribution over signals. 
It's crucial to note that even though the seller designs these experiments, the signals generated by them are privately observed by the buyers. This aligns with the application of private inspection, where the seller can restrict the method the buyers can use for inspecting the production, but the inspection outcome is only privately observed by the buyers. 
Moreover, in this model, the experiments are conducted independently for each buyer. This means that the seller cannot establish correlations between the signals received by different buyers, and what buyer $i$ learns is entirely unrelated to what buyer $i'$ learns about their private values. We provide additional justifications of our modeling choices in \cref{sec:prelim}.

Following the advertisement, the seller commits to a mechanism for selling the item. This chosen mechanism must satisfy two crucial constraints: it should be incentive compatible, meaning it incentivizes buyers to truthfully disclose their posterior values for the item based on the advertisement, and individually rational, ensuring that buyers have a non-negative expected utility when participating in the auction. 
The seller's objective is to jointly design the signal structures and a subsequent mechanism to maximize the expected payments collected from the buyers, subject to the incentive compatible and individually rational constraints. 

To illustrate the joint design problem and the challenges it presents, we begin with a simple example with two buyers. In this example, each buyer has a value that is uniformly and independently drawn between $1$ and $2$. These realized values are initially unknown to both buyers.
Consider a mechanism in which the values are fully disclosed to buyer 1 while providing no information to buyer 2. Following the disclosure of information, buyer 1 perfectly learns his value, while buyer 2 is aware only that her value lies uniformly between $1$ and $2$, thereby implying an average value of $1.5$, which is also her maximum willingness to pay for the item. Subsequently, the seller approaches the buyers sequentially, first posting a price of $2$ to buyer 1 and then a price of $1.5$ to buyer 2 if the item is not sold to buyer 1. As buyer 2 always purchases the item when offered, it is easy to verify that this mechanism is revenue optimal for the seller, as the expected revenue equals $1.75$, which coincides with the optimal welfare in this example.\footnote{Although the seller can extract the full surplus in this example, this observation does not hold in general. We provide an illustration for the failure of full surplus extraction in \cref{sub:example}.} 

Despite its relevance in a wide array of applications, our understanding of the joint design problem still lags behind that of the pure mechanism design problem. While \citet{bergemann2007information} offers important insights into the structure of the optimal solution, the computational complexity of the joint design problem remains unexplored. Our paper seeks to bridge this gap in understanding by investigating the following question: 
\begin{align*}\label{eq:question1}
&\emph{What is the computational complexity of the joint design problem?} \tag{*}
\end{align*}
An important observation from the aforementioned example is that the optimal solution can exhibit asymmetry, even in cases where the problem is  symmetric. This asymmetry reveals the non-concave nature of the joint design problem, indicating significant challenges in computing the optimal solution.

\subsection{Our Results and Techniques}\label{sec:results and techniques}

We first show that in contrast to auctions with exogenous signal structures where the optimal mechanism can be efficiently computed through linear programming \citep{myerson1981optimal,border1991implementation}, computing the optimal solution in the joint design problem is NP-hard.
The construction of the hard instances utilizes a reduction from the well-known NP-complete \textsc{Partition} problem \citep{garey1979computers}, which is inspired by \cite{xiao2020complexity} and \cite{chen2014complexity}.

Given the hardness result, we now focus on designing efficient algorithms for computing approximate optimal solutions. We begin by showing that a simple signal structure with binary signals, followed by the sequential posted pricing mechanism, can yield revenue that is at least a fraction of $1 - \sfrac{1}{e}$ of the optimal welfare. This result leverages the correlation gap technique \citep{agrawal2010correlation, yan2011mechanism}. Our result also implies that the worst-case multiplicative gap between optimal welfare and optimal revenue is at most $\sfrac{e}{(e-1)}$ for all valuation distributions, assuming the seller can design the signal structures. This finding sharply contrasts with scenarios where the seller is unable to design the signal structures, resulting in unbounded gaps between welfare and revenue.

Our main result is a PTAS for computing the signal structure and subsequent mechanism with at most an $\epsilon$ multiplicative loss in expected revenue compared to the optimal solution. To illustrate our algorithm, it is helpful to first explain some structures of the joint design problems. As the bidders have quasi-linear utilities, their willingness to pay after receiving a signal is simply the corresponding posterior mean. We refer to this distribution of the posterior mean as the value distribution induced by the signal structure. Once the signal structure is fixed, the subsequent optimal mechanism is Myerson's optimal auction with respect to the induced value distributions. 

\citet{bergemann2007information} offers an insightful characterization of the optimal signal structure, that is, it must be \emph{monotone}. However, their analysis does not preclude the use of \emph{randomized} signal structures, even when dealing with discrete value distributions. As a result, this implies that the number of potential signal structures is infinite. We provide a more refined characterization (\Cref{lem:partition}) that the optimal signal structure must be \emph{deterministic}, which we refer to \emph{monotone partitional signal structure}. Our new characterization significantly reduces the pool of potential candidates for the optimal signal structure to \(O\rbr{\exp(m)}\), where \(m\) is the size of the support of the value distribution. This refined characterization requires several new observations and is crucial for our PTAS.

Though our strengthened characterization reduces the number of possible signal structures per buyer from infinite to \(O\rbr{\exp(m)}\), it is still too expensive to exhaustively search over them. We next show how to combine dynamic programming with the problem's inherent structure to design a PTAS to compute the optimal monotone partitional signal structure. As shown by~\citet{bergemann2007information}, the value distribution induced by the optimal signal structure must be regular for every buyer. Our first step in designing the PTAS is to combine this observation with the revenue equivalence by \citet{myerson1981optimal} to convert the problem of revenue maximization to virtual welfare maximization. We further decompose the problem of virtual welfare maximization into the following two sub-problems:
\begin{enumerate}
    \item For each buyer \(i\), what  \emph{regular distributions} are inducible by a monotone partitional signal structure, and what are the corresponding distributions of virtual values?
    \item How to choose an inducible distribution of virtual value from each buyer so that the expected maximum of the virtual values is maximized?
\end{enumerate}

For the first sub-problem, our plan is to construct an $\varepsilon$-net over all \emph{regular} inducible distributions for each buyer. This requires constructing a set $\mathcal{F}$ that satisfies the following two conditions: (i) for every regular inducible distribution, there must be a distribution in $\mathcal{F}$ that is $\varepsilon$-close to it; and (ii) every distribution in $\mathcal{F}$ should be  $\varepsilon$-close to at least one regular inducible distribution. Here the $\epsilon$-distance is measured according to the differences in virtual welfares. Note that constructing an 
$\varepsilon$-net over all distributions inducible by a monotone partitional signal structure is relatively straightforward in the absence of the regularity condition. While the regularity condition allows us to effortlessly shift from revenue maximization to virtual value maximization, i.e., avoiding the need for ironing and ironed virtual welfare maximization, it also complicates the construction of the 
$\varepsilon$-net. To address this challenge, we provide a dynamic programming algorithm that memoizes carefully designed information at each state to efficiently construct such an $\varepsilon$-net, as detailed in Algorithm~\ref{alg:dps}.

For the second sub-problem, we employ another dynamic programming algorithm to keep track of the distribution of the maximum virtual value among the first $j$ buyers rather than the profile of these buyers' induced value distributions. This method substantially reduces the number of states we need to store, thereby speeding up the algorithm. It is important to note that in order to prevent an exponential increase in the number of states, it is necessary to round the distribution of the maximum virtual value at every step of the DP. This requirement introduces additional subtleties in both the design of the algorithm and its analysis.

\subsection{Related Work}
The papers most closely related to ours include those by \citet{bergemann2007information,Bergemann2022Optimal,bergemann2022screening}. In particular, \citet{bergemann2007information} and \citet{bergemann2022screening} show that the revenue-optimal signal structure consists of a finite partition of the type space. 
Our paper extends the characterizations in \citet{bergemann2007information}
by introducing a novel algorithm for efficiently computing a nearly optimal partition within polynomial time.


One line of work that relates to our paper considers the optimal signaling problem in Ad auctions \citep[e.g.,][]{emek2014signaling,cheng2015mixture,badanidiyuru2018targeting,bacchiocchi2022public}.
These papers consider the model where the seller is privately informed about the states for the Ad slots, and can design the signaling scheme for disclosing information about Ad slots to the buyers.  
After the disclosure, the Ad slots are sold via second-price auctions if there is a single slot or VCG mechanisms if there are multiple slots. 
\citet{cheng2015mixture} show that additive PTAS algorithms exist for selling a single slot in such environments.
If the selling mechanism is not restricted to second-price auctions, \citet{daskalakis2016does} show that the optimal joint design problem in Ad auctions corresponds to multi-item auctions with additive valuations.
In contrast, in our model, the seller is uninformative about the states or the signals sent to the buyers, and the selling mechanism is not restricted to second-price auctions. 
Moreover, from a technical perspective, our problem is non-concave, in contrast to ad auctions where the optimal joint design problem can be formulated as a linear program.
We provide a multiplicative PTAS algorithm for solving the optimal joint design of signal structure and selling mechanism in our model.

The computation of optimal advertisement has also been studied in \citet{zheng2021optimal} for selling information products. 
A significant distinction between selling information products and physical goods lies in the fact that advertising solely influences the buyers' posterior values for winning the physical goods. In contrast, when advertising information products, it also impacts the quantity of information that can be sold.

Our paper also contributes to the broader domain of information design \citep[e.g.,][]{rayo2010optimal,kamenica2011bayesian}. 
Additionally, there has been a growing body of research that studies the computational complexities of finding optimal signal structures across various settings \citep[e.g.,][]{dughmi2016algorithmic,xu2020tractability}.
Our paper extends this line of work by providing novel computational insights into optimal signal structures when the designer also has the capacity to jointly design the mechanisms for item sales.

Our paper relates to the computational complexity of finding the optimal mechanisms in classic auction environments without advertisements \citep[e.g.,][]{cai2015extreme,cai2012algorithmic,cai2012optimal,cai2013understanding,cai2013reducing,cai2013simple,chen2014complexity,chen2022complexity,papadimitriou2016complexity,alaei2019efficient}. 
For single-item auctions, the optimal mechanism can be computed efficient using the characterization in \citet{myerson1981optimal}. 
For multi-item auctions, \citet{cai2012optimal} provide efficient algorithms for computing the optimal mechanisms for additive buyers. 
When buyers have unit-demand valuations, \citet{chen2014complexity} show that computing the optimal pricing scheme is NP-complete if the valuation distribution on each item has support size at least 3.

Our paper has focused on a specific model concerning optimal information in auctions that fits the applications of optimal advertisements or private inspection. 
Note that there are other papers that explored various models for studying optimal information under different settings and objective functions. Although those models are suitable for studying other applications such as selling information, they do not fit into the applications we are interested in this paper. 
For instance, several papers have considered a model in which the principal can charge prices to the buyer even if the item remains unsold \citep[e.g.,][]{esHo2007optimal,li2017discriminatory,krahmer2020information,smolin2023disclosure}. In this scenario, the principal's goal is to design a pricing scheme that encompasses both the price for providing information and the price for selling the item, with the aim of maximizing the expected revenue.
Another stream of work concerning information in auctions explores signal structures designed to maximize the expected utilities of the participating buyers \citep{roesler2017buyer,deb2021multi,brooks2021optimal,chen2023information}. 
In this context, the designer publicly discloses a signal structure for the buyers to learn their values. Subsequently, the seller strategically responds by proposing a mechanism that optimizes her expected revenue. The key distinction here is that the designer's goal is to implement a signal structure that maximizes the expected utilities of the buyers while anticipating the mechanism adopted by the seller after the signal structure is publicly disclosed.
Lastly, our model also intersects with research on market segmentation and endogenous principal learning, where the focus is essentially on providing additional information to the principal about the buyers' values, rather than the buyers themselves \citep[e.g.,][]{bergemann2015limits,haghpanah2023pareto,CL-23}. 
In all these alternative models concerning optimal information in auctions, the primary objective is to provide qualitative insights into optimal signal structures. However, the computational aspects of these models have, to a large extent, been underexplored.

%% file: content/prelim.tex
\section{Preliminaries}
\label{sec:prelim}

We consider the problem of selling a single item to $n$ heterogeneous buyers. 
The value $\val_i$ of each buyer~$i$ is drawn independently from a commonly known distribution $\dist_i$
with finite support $\valspace_i\subseteq\reals_{+}$
and probability mass function $\density_i$. 
Let $m=\max_{i\in[n]}|\valspace_i|$ be the maximum cardinality of the supports
and let $\dists = \dist_1\times\cdots\times\dist_n$ be the joint distribution profile. 
In our model, all buyers have linear utilities. 
That is, 
the utility of buyer $i$ for receiving the item with probability $\alloc_i$ while paying price~$\pay_i$ is 
\begin{align*}
\util_i(\alloc_i,\pay_i; \val_i) = \val_i \cdot \alloc_i - \pay_i. 
\end{align*}
Sometimes we omit $\val_i$ in the notation when the value of the buyer is clear from context. 
The goal of the seller is to maximize the total revenue from the buyers, i.e., $\sum_{i} \pay_i$. 

\paragraph{Signal Structures}
Initially, neither the seller nor the buyers possess any information regarding the values of individual buyers beyond the prior distributions.
In our model, the seller can design signal structures, i.e., Blackwell experiments, that allow each buyer~$i$ to privately learn about his own value~$\val_i$. 
Specifically, the seller can commit to a signal structure $(\signals_i, \experiment_i)$ for each buyer $i$
where $\experiment_i: \valspace \to \Delta(\signals_i)$. 
Each buyer $i$ will privately observe the realized signal $\signal_i$
and update his belief according to the Bayes' rule.

In this paper, a family of signal structures that are of particular interest are monotone partitional signal structures. 
\begin{definition}[Partitional Signal Structures]
$(\signals_i, \experiment_i)$ is a \emph{partitional signal structure} 
if $\signals_i\subseteq 2^{\valspace_i}$ such that 
(1) $\signal_i\bigcap\signal'_i = \varnothing$ for any $\signal_i\neq\signal'_i \in S_i$;
(2) $\bigcup_{\signal_i\in\signals_i} \signal_i= \valspace_i$;
and (3) $\experiment_i(\val_i) = \signal_i$ if $\val_i\in\signal_i$.
Moreover, $(\signals_i, \experiment_i)$ is a \emph{monotone partitional signal structure} 
if it is a partitional signal structure and for any $\signal_i\neq\signal'_i$ and any $\val_i, \hat{\val}_i\in \signal_i, \val'_i\in\signal'_i$, 
$\val_i < \val'_i$ if and only if $\hat{\val}_i < \val'_i$.
\end{definition}

\paragraph{Mechanisms}
A mechanism $\mech=(\lbr{(\signals_i, \experiment_i)}_{i\in[n]}, \lbr{\alloc_i}_{i\in[n]},\lbr{\pay_i}_{i\in[n]})$ is a tuple containing 
a signal structure profile $\lbr{(\signals_i, \experiment_i)}_{i\in[n]}$, 
an allocation rule $\alloc_i : \signals_i \to [0,1]$
and a payment rule $\pay_i : \signals_i \to \reals$ for each buyer $i\in[n]$. 
Mechanism $\mech$ is 
\emph{Bayesian incentive compatible} (BIC) if 
\begin{align*}
\expect[\experiments,\dists]{\util_i(\alloc_i(\signal_i, \signal_{-i}), \pay_i(\signal_i, \signal_{-i}); \val_i) \given \signal_i} \geq 
\expect[\experiments,\dists]{\util_i(\alloc_i(\signal'_i, \signal_{-i}), \pay_i(\signal'_i, \signal_{-i}); \val_i) \given \signal_i}
\end{align*}
for all $\signal_i,\signal'_i\in\signals_i$,
where $\expect[\experiments,\dists]{\cdot}$ indicates that all values are generated according to $\dists$ and all signals are generated according to signal structures~$\experiments$.
Moreover, this mechanism is \emph{interim individual rational} (IIR) if
\begin{align*}
\expect[\experiments,\dists]{\util_i(\alloc_i(\signal_i, \signal_{-i}), \pay_i(\signal_i, \signal_{-i}); \val_i) \given \signal_i} \geq 0
\end{align*}
for all $\signal_i\in\signals_i$. 
To simplify the notations, we also omit $\dists$ in $\expect[\experiments,\dists]{\cdot}$
when it is clear from context. 

For any BIC-IIR mechanism $\mech$, we denote its expected revenue as $\rev(\mech) 
= \expect[\experiments,\dists]{\sum_{i\in[n]}\pay_i(\signal_i, \signal_{-i})}$
and its expected welfare as $\wel(\mech) = \expect[\experiments,\dists]{\sum_{i\in[n]}\val_i\cdot\alloc_i(\signal)}$. 


\paragraph{Discussion of Modeling Choices}
In our paper, we imposed the restriction that the signals from different buyers are independent. This reflects scenarios such as advertising or private inspections, where buyers are provided with the opportunities to learn their own values. In these applications, how buyer $i$ learns their value should not be contingent on how buyer $j\neq i$ learns theirs.\footnote{Naturally, if the buyers' true values for the items are correlated, their signals can also be correlated through their correlation to the true values. 
In environments with correlated values, \citet{cremer1988full} show that the seller can extract the full surplus by providing full information to the buyers. 
In our paper, we focus on environments with independent values and hence independent signals. 
}

In addition, we focus on mechanisms that provide information for free, i.e., advertisements or inspection opportunities come at no cost to the buyers. This aligns with applications where charging prices for provided information is generally not adopted in practice. 
If the seller were permitted to levy additional fees for providing information, similar to the literature on auctions with endogenous entry \citep{menezes2000auctions}, the seller could extract the entire surplus through a mechanism conducting a second-price auction followed by full information disclosure. In such a mechanism, each buyer would be charged their expected utility for receiving information and participating in the auction.

Finally, we allow the seller to use arbitrary information structure for disclosing information to individual buyers. 
While such complex structure may seem impractical for real-world applications, we show that the optimal mechanism as well as the mechanisms computed by our polynomial time algorithms only use \emph{partitional information structures}. In a partitional information structure, the signals reveal to the buyers the range of their values for the item, a feature we consider to be  implementable in real-world settings. Moreover, our constant factor approximation only uses binary signals, which further enhances its practical appeal.

\paragraph{Optimal $k$-Signal Problem} In this paper, we mainly focus on the following computational problem called \probkname problem. 
Specifically, let $\mechs_k$ be the set of BIC-IIR mechanisms in which the cardinality of each signal space is at most $k$, 
i.e., $|\signals_i|\leq k$ for all $i$. 
Given any prior distributions $\dist$ over values and any positive integer $k\geq 2$, the objective of the seller is to design an algorithm that computes a mechanism $\mech^*_k$ in $\mechs_k$ that maximizes the expected revenue, i.e., $\mech^*_k = \argmax_{\mech\in\mechs_k}\; \rev(\mech)$. 
To simplify the analysis in later sections, let \begin{align*}
    \optwel_k = \max_{\mech\in\mechs_k} \wel(\mech)\qquad \text{ and } \qquad\optrev_k = \max_{\mech\in\mechs_k} \rev(\mech)
\end{align*}
be the optimal welfare and optimal revenue respectively with at most $k$ signals for each buyer. 

In our paper, we also consider the problem where the cardinality of the signal space is unrestricted. 
We omit $k$ in the subscripts from notations if the signal space is a compact set without cardinality constraints, and we refer to the optimal design problem without cardinality constraints as the \probname problem.

\subsection{Illustrative Example}
\label{sub:example}
Before going into the details of our analysis, we provide a simple example to illustrate the joint design problem. 

Consider the example with two buyers. In this example, each buyer's value is uniformly and independently drawn among 0, 1 and 2. 
According to \cref{lem:partition}, which we will show later in \cref{sec:nph}, the optimal signal structure is a monotone partitional signal structure. By enumerating all possible monotone partitions, we show that for each buyer $i\in{1,2}$, the optimal signal structure $(\signals_i,\experiment_i)$ comprises a binary signal space $\signals_i = {h,l}$, where signal~$h$ is sent if the value is either 1 or 2, and signal $l$ is sent otherwise.

Conditional on receiving signal $h$, by Bayes' rule, the buyer's posterior belief about their own value is a uniform distribution between 1 and 2, with an expected value of 1.5. Therefore, the buyer's expected utility for winning the item conditional on receiving signal $h$ is $1.5$ and hence his maximum willingness to pay in this case is also $1.5$. Similarly, conditional on receiving signal~$l$, the buyer's posterior belief about their own value is a point mass distribution at $0$.
With this signal structure, the optimal selling mechanism that ensures both incentive compatibility and individual rationality is to sequentially offer each buyer a take-it-or-leave-it price of 1.5.
The item is sold if at least one of the buyers receives signal $h$, which occurs with a probability of $1-\frac{1}{3}\times\frac{1}{3}=\frac{8}{9}$. Therefore, the expected revenue in the optimal mechanism is $\frac{4}{3}$. 

In this example, the optimal welfare is $\frac{13}{9}$, which is strictly greater than the optimal revenue. Therefore, full surplus extraction is generally not possible, even when the seller can design the signal structures.



%% file: content/nph.tex
\section{NP-hardness}
\label{sec:nph}

In this section, we prove that both the \probname problem and the \probkname problem for any $k\geq 2$ are NP-hard {when the maximum size of support is at least \(3\)}.  If the maximum support size is at most \(2\), we demonstrate in \Cref{apx:nph} that the problem can be solved within polynomial time.
The missing proofs in this section are provided in \cref{apx:nph}.


\begin{theorem}[NP-hardness]
\label{thm:nph}
Both the \probname problem and the \probkname problem for any $k\geq 2$ are NP-hard when \(\max|V_i| \geq 3\).
\end{theorem}
We divide the proof of \cref{thm:nph} into two parts. 
First we show that the \probname problem is NP-hard even if the valuation distributions have support size at most 3. 
\begin{lemma}\label{lem:nph_support3}
The \probname problem is NP-hard even when $|\valspace_i| \leq 3$ for all $i$.
\end{lemma}





The proof of \cref{lem:nph_support3} relies on the following characterization of the optimal signal structure. 
\begin{lemma}\label{lem:partition}
For both the \probkname and the  \probname problem, there exists an optimal mechanism in which the signal structure $(\signals_i,\experiment_i)$ is a monotone partitional signal structure for all buyer $i$. 
\end{lemma}
\citet{bergemann2007information} consider a model with continuous distributions and show that it is without loss to focus on monotone partitional signal structures. 
In the case of discrete distributions, their result translates into the optimality of randomized monotone partitional signal structures where a value can be randomly assigned to two different partitions. 
\cref{lem:partition} further shows that such \emph{randomization is unnecessary} to ensure the optimality.

An immediate corollary of \cref{lem:partition} is that the cardinality of the signal space for each buyer~$i$ in the optimal mechanism is upper-bounded by the cardinality of the support of the valuation distribution. Therefore, for instances where the valuation distributions have a support size of at most 3, the cardinality of the each buyer's signal space in the optimal mechanism is at most 3. {This strengthened characterization is crucial for our PTAS, as it provides a polynomial upper bound on the support size of the buyers' induced value distributions. This bound is key to making the dynamic programming feasible.
The characterization also has some interesting implications for the hardness results. In particular, it means} that the constraints on the cardinalities of the signal spaces are not binding for those instances when $k\geq 3$. 
Hence, \cref{lem:nph_support3} also implies that the \probkname problem is NP-hard for any $k\geq 3$.

{In the following, we present the hard instance used to derive \Cref{lem:nph_support3} and a brief description of our reduction. Our proof is inspired by \citet{xiao2020complexity} and \citet{chen2014complexity}. It is derived through a reduction from the well-known NP-complete \textsc{Partition} problem \citep{garey1979computers}:

\textsc{Partition:} Given \(n\) positive integers \(c_1, c_2, \cdots, c_n\), is it possible to find a subset \(S\subseteq [n]\) such that \(\sum_{i\in S} c_i = \sum_{i\notin S} c_i\)?

Given an instance of \textsc{Partition}  parameterized by \(c_1, c_2, \cdots, c_n\), it is without loss of generality to assume \(c_1 \geq c_2 \geq \cdots \geq c_n\). We define \(M = c_1\cdot 2^n\) and \(t_i = \half \sum_{\tl{j\in[n]}{j\neq i}} c_j / M\). For each \(i \in [n]\), let \(r_i\) and \(q_i\) represent the unique, non-negative solutions to the following system of equations:
\begin{align*}
r_i + q_i &= c_i / M.\\ 
2q_i + r_i t_i &= q_i + r_i \label{eq:Support3NPH2}
\end{align*}

That is to say,
\begin{align*}
r_i = \frac{c_i}{M(2-t_i)} 
\quad\text{and}\quad
q_i = \frac{c_i(1-t_i)}{M(2-t_i)}.
\end{align*}

We now construct the following hard instance for the \probname problem comprising \(n + 1\) buyers. Specifically, for the \(i\)-th buyer among the first \(n\) buyers, the value distribution \(F_i\) is defined as follows:
\begin{align*}
    v_i \sim F_i, v_i = \begin{cases}
        3 & \text{w.p. }  q_i\\
        t_i + 1 & \text{w.p. } r_i\\
        0 & \text{w.p. } 1 - q_i - r_i
    \end{cases}
\end{align*}
}

In the constructed hard instance for the \probname problem with value support size 3, for each buyer $i$, one of the possible values is $0$, which without loss of optimality separated from the rest of the values in the optimal signal structure. Moreover, for the remaining two values in the support, since the optimal signal structure is monotone partitional, it can only either full disclose these two values or pool them together as a single signal. 
Essentially this provides a natural partition of the buyers based on whether their high values are fully disclosed, and we constructed the values of the items in the \probname problem based on any given \textsc{Partition} problem instance such that the expected revenue is maximized if and only if this partition of buyers is an optimal solution to the \textsc{Partition} problem.

Next, it is sufficient to focus on the \probkname problem with $k=2$. The high level intuition is similar to \cref{lem:nph_support3}. The difference is that here we will provide a reduction from the NP-complete problem \textsc{Subset Product} \citep{ng2010product}.

\begin{lemma}\label{lem:nph_2signal}
The \textsc{Optimal 2-Signal} problem is NP-hard {when \(\max |V_i| \geq 3\)}. 
\end{lemma}

{
We again offer a high-level overview for the proof of \cref{lem:nph_2signal} with the full details provided in \cref{apx:nph}. The hard instance in this lemma is inspired by \citet{agrawal2020optimal}. We provide the reduction from the following NP-complete \textsc{Subset Product} problem \citep{ng2010product}:

\textsc{Subset Product:} Given integers \(a_1, a_2, \cdots, a_n\) with \(a_i > 1\) and a positive integer \(B\), is it possible to find a subset \(S\subseteq [n]\) such that \(\prod_{i\in S} a_i = B\)?

Given $a_1, a_2, \ldots, a_n$ and $B$, we construct a hard instance of the \textsc{Optimal 2-Signal} problem with \( n \) buyers.  For each buyer \( i \), we define the value distribution \( F_i \) such that:
\begin{align*}
    v_i \sim F_i, v_i = \begin{cases}
        1 & \text{w.p. }  \frac{a_i - 1}{a_i}\\
        \frac{B^2 - a_i}{B^2 + 1} & \text{w.p. } \frac{a_i - 1}{a_i^2}\\
        0 & \text{w.p. } \frac{1}{a_i^2}
    \end{cases}
\end{align*}

In our construction, \(0\) is once again included in the distribution's support. Given that no more than 
\(2\) signals are permitted, it can be shown that there are only two possible optimal signal structures: either pooling the highest two values together and leaving \(0\) alone, or pooling the lowest two values together and isolating~\(1\). Note that in both signal structures, only the higher signal possesses a positive virtual value. Moreover, the positive virtual values in both signal structures are constants independent from \(i\). Therefore, finding the optimal \(2\)-signal structure in this instance is equivalent to deciding a set of buyers \(T\subseteq [n]\) who pool the highest two values together. It turns out that the revenue of any signal structure is a function of \(\prod_{i\in T} a_i\) and attains the maximum when \(\prod_{i\in T} a_i = B\).
}

\cref{thm:nph} holds immediately by combining \cref{lem:nph_support3} and \cref{lem:nph_2signal}.

%% file: content/algorithms.tex
\section{Polynomial Time Algorithms}
\label{sec:algorithms}

In this section, we provide several polynomial time algorithms for approximating the optimal revenue of the optimal signal structure with at most $k$ signals.
The missing proofs in this section are provided in \cref{apx:algorithms}.


\subsection{Warm-up: A Simple Constant Approximation Algorithm}
\label{sub:approx}

We first introduce a simple $\sfrac{e}{(e-1)}$-approximate algorithm for this problem.

\begin{theorem}[Welfare Approximating Algorithm]
\label{thm:simplealgo}
For any $k\geq 2$, there exists an algorithm for the \probkname problem with running time $\poly(n,m)$ 
that computes a mechanism $\mech$ with expected revenue at least $1-\sfrac{1}{e}$ fraction of the optimal welfare, 
i.e., $\rev(\mech) \geq (1-\sfrac{1}{e}) \cdot \optwel$. 
\end{theorem}
Since welfare is an upper bound on the optimal revenue and the optimal revenue is higher when the cardinality constraints on the signal space is relaxed, we have 
\begin{align*}
\rev(\mech)\geq \rbr{1 - \frac1e} \optwel \geq \rbr{1 - \frac1e} \optrev\geq \rbr{1 - \frac1e} \optrev_k.
\end{align*}
Therefore, the constructed mechanism given in \cref{thm:simplealgo} is an $\sfrac{e}{(e-1)}$ approximation to the optimal revenue. 
Moreover, 
An immediate corollary of \cref{thm:simplealgo} is that the optimal revenue is a constant approximation to the optimal welfare when the seller can design the signal structures for arbitrary buyer distributions~$\dists$. In contrast, this claim is no longer valid even when buyers have regular value distributions if the seller lacks the ability to design the signal structures.
\begin{corollary}\label{cor:rev_approx_wel}
$\optrev \geq (1-\sfrac{1}{e}) \cdot \optwel$. 
\end{corollary}

\begin{algorithm}[t]
\caption{A constant approximation algorithm}
\label{alg:constant_approx}
\KwIn{Prior distributions $\dists = \dist_1\times\cdots\times \dist_n$.}
\KwOut{Mechanism $\mech=(\lbr{(\signals_i, \experiment_i)}_{i\in[n]}, \lbr{\alloc_i}_{i\in[n]},\lbr{\pay_i}_{i\in[n]})$.}

Computes the probability that $i$ has the largest value among all buyers as
\[q_i = \prob[\vals\sim \dists]{\val_i = \argmax_{j\in[n]} \val_j}\]
where ties are broken according to lexicographic order. 

Let $\bar{\val}_i$ be the minimum value such that $\dist_i(\bar{\val}_i) \geq 1-q_i$ and let 
$d_i = \frac{\dist_i(\bar{\val}_i) - 1+q_i}{\density_i(\bar{\val}_i)}$. 

\textbf{Signal Structure:} For each buyer $i$, consider a binary signal space $\signals_i = \{0,1\}$ and signal structure $\experiment_i$ such that 
\begin{align*}
\experiment_i(\val_i) = \begin{cases}
\delta_1 & \val_i > \bar{\val}_i\\
d_i\delta_1 + (1-d_i)\delta_0 & \val_i = \bar{\val}_i\\
\delta_0 & \val_i < \bar{\val}_i
\end{cases}
\end{align*}
where $\delta_{\signal}$ is a point mass measure on signal $\signal$ for all $\signal\in\signals$. 

Let $b_i = \expect[\experiment_i, \dist_i]{\val_i\given \signal_i = 1}$
be the expected value of buyer $i$ conditional on receiving signal $1$. 

\textbf{Selling Mechanism:} In mechanism $\mech$, the seller uses sequential posted pricing by approaching the buyers according to an order that is decreasing in $b_i$. 
That is, for each buyer $i$, for any signal profile $\signal$, $\alloc_i(\signal) = 1$ if buyer $i$ is the first buyer (according to decreasing order on $b_i$) such that $\signal_i = 1$. 
Moreover, $\pay_i(\signal) = b_i$ if $\alloc_i(\signal) = 1$.
\end{algorithm}

The algorithm for \cref{thm:simplealgo} is constructed in Algorithm \ref{alg:constant_approx}.
Intuitively, the algorithm first computes the probability $q_i$ each buyer is the highest value buyer given the prior distribution, 
and sends a binary signal to the buyer to indicate whether his value is in the upper $q_i$ quantile. 
Then the algorithm computes the expected value $b_i$ of each buyer conditional on receiving the high value signal.
The seller approaches the buyers sequentially according to the decreasing order on conditional values~$b_i$, and offers prices that equal the conditional values. 
It is easy to verify that the algorithm can be implemented in $\poly\rbr{n,m}$ time. 
Therefore, it is sufficient to show that the revenue derived from the constructed mechanism is at least $1-\sfrac{1}{e}$ fraction of the optimal welfare, which can be proved by applying the technique of correlation gap \citet{yan2011mechanism}. 
\begin{lemma}\label{lem:simplepprev}
The revenue of mechanism $\mech$ constructed in Algorithm \ref{alg:constant_approx}
is at least $1-\sfrac1e$ fraction of the optimal welfare, i.e., 
$\rev(\mech)\geq \rbr{1 - \sfrac1e} \optwel.$
\end{lemma}
\begin{proof}
The proof here is an application of correlation gap. We define a unit-demand set function $f(S)$ for $S\subseteq [n]$ where
\[f(S) = \max_{i\in S} b_i.\]

It is clear to verify that $f(S)$ is a submodular function. For the sequential posted price mechanism, we define $\mathcal{A}$ as the set of buyers that receive signal $1$. Therefore, the revenue of mechanism $\mech$ is the expectation of $f(\mathcal{A})$:
\begin{equation}
\label{eq:simplerevproof1}
\rev\rbr{\mech} = \expect[\mathcal{A}]{f(\mathcal{A})}.
\end{equation}

Since $q_i$ is the probability that buyer $i$ has the highest value among all buyers, and $b_i$ is the expectation of $\val_i$ in the top $q_i$ quantile, it follows that
\[\expect{\val_i \given \val_i = \argmax_{j\in[n]} v_j} \leq b_i.\]

Therefore, by ex ante relaxation, we have
\begin{align}
\label{eq:simplerevproof2}
\begin{split}
    \optwel = \sum_{i=1}^n \prob{\val_i = \argmax_{j\in[n]} v_j} \expect{\val_i \given \val_i = \argmax_{j\in[n]} v_j}
    \leq \sum_{i=1}^n q_i b_i.
\end{split}
\end{align}

Note that for each buyer $i$, it is in the set $\mathcal{A}$ with probability $q_i$ independently. Now consider another set $\mathcal{B}$ where $\mathcal{B} = \left\{i\given \val_i = \argmax_{j\in[n]} v_j\right\}$. It is clear that the marginal probability that $i$ is in the set $\mathcal{B}$ is also $q_i$. However, the events that whether buyer $i$ is in the set $\mathcal{B}$ are not independent. Since the correlation gap for submodular functions is $\frac{e}{e-1}$ \citep[see][]{yan2011mechanism}, it holds that
\[\expect[\mathcal{A}]{f(\mathcal{A})} \geq \left(1-\frac1e\right)\expect[\mathcal{B}]{f\rbr{\mathcal{B}}}. \]
Follow from the definition of $\mathcal{B}$, it holds that 
$\expect[\mathcal{B}]{f\rbr{\mathcal{B}}} = \sum_{i=1}^n q_i b_i$.
Combining the two inequalities above, we have
\begin{align}
\label{eq:simplerevproof3}
\begin{split}
    \expect[\mathcal{A}]{f(\mathcal{A})} \geq \left(1-\frac1e\right) \cdot \sum_{i=1}^n q_i b_i.
\end{split}
\end{align}
Finally, \cref{lem:simplepprev} holds by combining inequalities~\eqref{eq:simplerevproof1}, \eqref{eq:simplerevproof2} and \eqref{eq:simplerevproof3}.
\end{proof}

\subsection{PTAS}
\label{sub:ptas}

In this section, we present a PTAS for the \probkname problem. The missing proofs in this section are provided in \cref{apx:ptas}.

\begin{theorem}[PTAS]
\label{thm:ptas}
For any $k\geq 2$ and  any $\epsilon > 0$,  there exists an algorithm that computes a multiplicative $(1-\varepsilon)$-optimal solution to the \probkname problem in time $\poly\rbr{\left(\frac{nm}{\epsilon}\right)^{f\left(\frac{1}{\varepsilon}\right)},n,m,k}$, where $n$ is the number of bidders, $m$ is the maximum cardinality of the support of the bidders' value distributions, and $f(\cdot)$ is a fixed function that is independent of the parameters of the \probkname problem.
\end{theorem}

Note that by \cref{lem:partition}, the maximum cardinality of the signal spaces is at most $m$ in the \probname problem. Therefore, by applying \Cref{thm:ptas} to the case where $k=m$, the resulting mechanism is approximately optimal for the \probname problem, even when the signal spaces are unrestricted. Hence, the algorithm proposed in \cref{thm:ptas} also serves as a PTAS for the \probname problem.
\begin{corollary}
\label{cor:ptas_unrestricted}
For any $\varepsilon>0$, there exists an algorithm that computes a multiplicative $(1-\varepsilon)$-optimal solution to the \probname problem in time $\poly\rbr{(\frac{nm}{\epsilon})^{f(\frac{1}{\varepsilon})},n,m}$, where $n$ is the number of bidders, $m$ is the maximum cardinality of the support of the bidders' value distributions, and $f(\cdot)$ is a fixed function that is independent of the parameters of the \probname problem.
\end{corollary}

To prove \cref{thm:ptas}, first note that it is without loss of generality to normalize the valuation distributions such that the optimal welfare equals $1$. 
We will show that Algorithm \ref{alg:ptas} outputs a mechanism with an $O(\epsilon)$ additive loss in expected revenue after the normalization of the distributions. 
Since \cref{cor:rev_approx_wel} implies that the optimal revenue is a constant approximation to the optimal welfare, any $O(\epsilon)$ additive loss immediately implies an $O(\epsilon)$ multiplicative loss in expected revenue.

The computation of the optimal signal structure relies on the revenue equivalence result in \citet{myerson1981optimal}, which converts the optimal revenue to the expected virtual surplus. 
\begin{definition}[Virtual Values]
\label{def:virtualvalues}
For any discrete distribution $\dist$ with probability mass function $\density$, for any valuation $\val$ in the support of $\dist$, letting $\val^+$ be the minimum value in the support of $\dist$ that is strictly larger than $\val$, 
the virtual value of $\val$ in distribution $\dist$ is 
\begin{align*}
\virtual(\val) = \val - \frac{(1-\dist(\val))(\val^+ - \val)}{\density(\val)}.
\end{align*}

We might, for the sake of convenience, extend the notation to use \(\virtual(\signal_i)\), where \(\signal_i\) represents a signal within the signal structure \(\left(\signals_i, \experiment_i\right)\). This is intended to denote the virtual value of signal \(s_i\) with respect to the distribution generated by the signal structure.

\end{definition}


\begin{definition}[Regularity]
A valuation distribution $\dist$ is \emph{regular} if the corresponding virtual value function $\virtual(\val)$ is non-decreasing in $\val$. 
\end{definition}

\begin{lemma}[Revenue Equivalence \citep{myerson1981optimal}]
\label{lem:rev_equiv}
For any regular distribution profile $\dist_1, \dots, \dist_n$, the optimal revenue equals virtual welfare, i.e., 
\begin{align*}
\optrev = \expect[\val\sim\dist]{\max_{i\in[n]} \virtual_i(\val_i)}.
\end{align*}
\end{lemma}

Moreover, Lemma 2 of \citet{bergemann2007information} show that it is without loss to consider monotone partitional signal structure such that the induced posterior valuation distribution is regular. 
Therefore, by converting the problem of revenue maximization into virtual welfare maximization, 
the high level idea of the PTAS is to discretize the space of virtual value distributions and use dynamic programming to keep track of the optimal virtual welfare the seller can collect from first $i$ buyers for any $i\leq n$. 
Formally, our PTAS consists of two phases:
\begin{itemize}
    \item \textbf{Phase 1:} Identify a monotone partitional signal structure that induces a \emph{regular} posterior valuation distribution, and achieves virtual welfare that is an $(1 - \varepsilon)$-approximation to the optimal revenue.

    \item \textbf{Phase 2:} Compute the subsequent selling mechanism that maximizes the expected revenue given the signal structure in Phase 1. 
\end{itemize}
The computation of Phase 2 can be solved in polynomial time by \citet{myerson1981optimal}. Therefore, in the rest of the section, we will focus on the computation of the optimal signal structure in Phase 1. 
We will first show how to discretize the space of virtual values, then illustrate how to compute the optimum using dynamic programming. 

\subsubsection{Discretization} 
We first describe our discretization scheme. 
For any buyer $i\in[n]$, let $\virtualdist_i$ be the distribution over virtual values and let $\virtualdensity_i$ be the probability mass function for virtual values. 
Given any virtual value distribution $\virtualdist_i$, for high virtual values above the threshold of $\epsilon^{-1}$, we store the contribution of those virtual values in a compensation term $\compensate_i$. 
For low virtual values below the threshold of $\epsilon^{-1}$, we simply discretize both the virtual values space and the probability space, and round down the distribution to the discretization grid. Specifically,

\begin{enumerate}
\item \textbf{Bounded Support:} For any virtual value distribution $\virtualdist_i$ with probability mass function $\virtualdensity_i$, we neglect the virtual values above $\epsilon^{-1}$, and for any $\val_i \in (0, \varepsilon^{-1}]$, we round it down to the nearest multiples of $\varepsilon$. 
Moreover, we round down the probability mass $\virtualdensity_i(\val_i)$ to the nearest multiples of $\frac{\varepsilon^4}{nm}$. 
That is, letting 
\begin{align*}
\grid = \lbr{z\varepsilon \mid z \in \naturals, z\varepsilon \leq \varepsilon^{-1}}
\end{align*}
be the discretization grid in the virtual value space, 
the discretized distribution $\hat{\virtualdist}_i$
has support in~$\grid$, and for any $\val_i \in \grid$, 
\begin{align*}
\hat{\virtualdensity}_i(\val_i) = \sum_{\val'_i \in [\val_i, \val_i+\epsilon)} \frac{z_{\val'_i} \epsilon^4}{nm}
\text{ where }
z_{\val'_i} = \max_{z\in\naturals} \lbr{ z \,\Big\vert\, \frac{z\varepsilon^4}{nm} \leq \virtualdensity_i(\val'_i)}.
\end{align*}
The rest of the probability mass in $\hat{\virtualdist}_i$ is reallocated to value $0$. 

\item \textbf{Compensation Term:} 
We define a compensation term $\compensate_i$ to preserve the mean for values above the threshold of $\varepsilon^{-1}$ as 
\begin{align}\label{eq:highvals0}
    \compensate_i = \expect[\val_i \sim \virtualdist_i]{\val_i \cdot \indicate{v_i > \varepsilon^{-1}}}.
\end{align}
Note that $\compensate_i\leq 1$ for all $i$ since the virtual welfare is at most the welfare, where the latter is normalized to 1. 
In the discretization scheme, we round down $\val_i \cdot\virtualdensity_i(\val_i)$ to the nearest multiples of $\frac{\epsilon}{nm}$ for any $\val_i > \epsilon^{-1}$ and then take the sum. 
\end{enumerate}
In the discretization scheme, we define $\discretize_V$ as the process of only discretizing the virtual value space, and $\discretize_Q$ as the process of only discretizing the probability space. 
Let $\discretize$ be the joint process of discretizing both virtual values and probabilities. 
Note that, in probability space discretization, even if multiple virtual values share the same discretization, we discretize the probabilities separately before aggregating them into a single virtual value in the discretized distribution. This procedure results in a larger discretization error, necessitating a finer discretization grid to maintain a small error. However, this sacrifice is crucial for the polynomial time computation of the set of implementable discretized distributions, as defined later in \cref{def:implementable}, using dynamic programming via Algorithm \ref{alg:dps}.
Similarly, we employ such discretization scheme for the compensation term $\compensate_i$ as well.

Given the discretization scheme, we can then define the space of feasible discretized virtual value distributions and compensation terms. 

\begin{definition}[Feasible Discretization Pairs]
A pair of distribution and compensation term $(\virtualdist, \compensate)$ is \emph{feasible} if 
\begin{enumerate}
\item The support of $\virtualdist$ is a subset of $\grid$, 
and the probability mass $\virtualdensity(\val)$ is an integral multiple of $\frac{\varepsilon^4}{nm}$ for any~$\val$.
\item $\compensate$ is an integral multiple of $\frac{\epsilon}{nm}$ and does not exceed $1$.
\end{enumerate}
\end{definition}
We denote $\feasibles$ as the set of feasible pairs. By slightly overloading the notation, sometimes we say a distribution $\virtualdist\in\feasibles$ if there exists $\compensate$ such that $(\virtualdist,\compensate)\in\feasibles$.

\begin{lemma}\label{lem:numberofdis}
The cardinality of the set of feasible pairs $\feasibles$ is $O\rbr{\rbr{\frac{nm}{\epsilon}}^{f\rbr{\frac{1}{\varepsilon}}}}$, where $f\rbr{\frac{1}{\varepsilon}}$ is a function that exclusively depends on $\frac{1}{\varepsilon}$.
\end{lemma}

Given the valuation distribution $\dist_i$, not all feasible pairs can be attained by implementing a monotone partitional signal structure with regular posterior value distribution, and then discretizing the virtual value distribution. We define the implementable set $\sets_i(\dist_i)$ as the set of implementable pairs derived based on value distribution $\dist_i$. 
Note that $|\sets_i(\dist_i)| \leq |\feasibles|$ for all $i$ and $\dist_i$. We often omit $\dist_i$ from the notation when it is clear from context. 


\begin{definition}[Implementable Discretization Pairs]
\label{def:implementable}
A pair of distribution and compensation term $(\virtualdist, \compensate)$ is \emph{implementable} given $\dist$, i.e., $(\virtualdist, \compensate)\in\sets(\dist)$, if there exists a monotone partitional $k$-signal structure $\rbr{\signals,\experiment}$ such that 
the induces posterior valuation distribution $\hat{\dist}$ is regular, and the discretized pair for the virtual value distribution of $\hat{\dist}$ matches $(\virtualdist, \compensate)$.
\end{definition}

\subsubsection{Dynamic Programming Algorithm}
Our algorithm has two main steps for computing the optimal monotone partitional $k$-signal structure. 
\begin{enumerate}
\item Computing the set of implementable pairs $\sets_i(\dist_i)$ for all buyer $i\in[n]$. 

\item Computing an implementable pair $(\virtualdist_i, \compensate_i) \in \sets_i(\dist_i)$ for each buyer $i$ to maximize the virtual welfare.
\end{enumerate}
Both steps are implemented using dynamic programming. 

\paragraph{Computation of Implementable Pairs}
For any buyer $i$, let \(\val_{i,j}\) be the \(j\)-th smallest value in the support of buyer~$i$'s valuation distribution. We use $m_i$ to denote $|\valspace_i|$.

We define the state for dynamic programming as \(\dps\rbr{k', j, j', \virtualdist_i, \compensate_i}\), which memorizes the following information as a tuple. The first component of the tuple is binary and is true if and only if the following two conditions are met: 
\begin{enumerate}[(i)]
\item There exists a monotone partitional signal structure $(\signals_i,\experiment_i)$ for values \( \left\{\val_{i,j}, \val_{i, j+1},\cdots, \val_{i, m_i}\right\}\) within distribution \(\virtualdist_i\) such that 
(1) the cardinality of the signal space \(|\signals_i|=k'\leq k\);
(2) the posterior value distribution $\hat{\dist}_i$ induced by $(\signals_i,\experiment_i)$ is regular; and 
(3) the discretized pair of the virtual value distribution $\virtualdist_i$ given $\hat{\dist}_i$ coincides with \( \rbr{\virtualdist_i, \compensate_i} \). 
Note that the signal structure for \( \lbr{\val_{i,1},\val_{i,2},\cdots, \val_{i,j-1}}\) remains undetermined, and hence we fulfill the probability mass of those part with a value of $0$ in posterior value distribution $\hat{\dist}_i$;

\item The partitional signal structure in (i) contains \( \lbr{v_{i,j}, v_{i,j+1},\cdots,v_{i,j'}} \) as a single signal. 
\end{enumerate}
The second component of the tuple stores the the maximum virtual value achievable by pooling valuations from $v_{i,j}$ to $v_{i,j'}$ subject to the constraint that the virtual values are monotone. 
The third component is used to store the index of the dynamic programming states for backtracking the implementable signal structures. 

The states of the dynamic programming are computed via Algorithm \ref{alg:dps}.
In Step 14, the function $\discretize$ is the discretization scheme that rounds down both distribution $\virtualdist'$ and the compensation term $\compensate'$.
Intuitively, Algorithm \ref{alg:dps} keeps track of all implementable pairs for using monotone partitions over values above $\val_{i,j}$. The algorithm repeatedly  incorporates new implementable pairs through the introduction of an additional partition that terminates at the value~$\val_{i,j-1}$. This addition must satisfy that the resulting virtual valuation is monotone.

With the states of the dynamic programming established, the computation of the sets of implementable pairs \(\sets_i\) becomes straightforward. A pair \(\left(\virtualdist_i, \compensate_i\right)\) is deemed implementable, that is, included in the set \(\sets_i\), if and only if there exists \(k' \leq k\) and \(j \geq 1\) such that the state \(\dps\left(k',1, j,\virtualdist_i,\compensate_i\right)\) evaluates to be true. 
Consequently, our approach involves enumerating all such pairs and verifying whether each candidate pair \(\rbr{\virtualdist_i, \compensate_i}\) qualifies for inclusion in \(\sets_i\). This process is efficiently executed in $\poly\rbr{\rbr{\frac{nm}{\varepsilon}}^{f\left(\frac{1}{\varepsilon}\right)},m,k}$ time.
The correctness of Algorithm~\ref{alg:dps} is established in the following lemma, with proof provided in \cref{apx:correctness}.

\begin{lemma}\label{lem:Algorithm2CorrectNess}
Algorithm \ref{alg:dps} correctly computes the set of implementable pairs $\sets_i$ for all $i\in[n]$, i.e., any feasible discretization pair $(\virtualdist'_i, \compensate'_i)$ is in $\sets_i$ if and only if there exists a partitioned signal structure \( (\signals_i, \experiment_i) \) that induces a regular virtual value distribution $\virtualdist$ such that $\discretize(\virtualdist)=(\virtualdist'_i, \compensate'_i)$. 

Additionally, for any $(\virtualdist'_i, \compensate'_i)\in\sets_i$
backtracking can compute a monotone partitioned signal structure $(\signals_i, \experiment_i)$ that induces a regular virtual value distribution $\virtualdist$ such that $\discretize(\virtualdist)=(\virtualdist'_i, \compensate'_i)$. 
This computation of both Algorithm \ref{alg:dps} and backtracking can be achieved in \(\poly\left(\left(\frac{nm}{\epsilon}\right)^{f\left(\frac{1}{\varepsilon}\right)}, n, m, k\right)\) time.
\end{lemma}

\paragraph{Virtual Welfare Maximization}
Given the set of implementable pairs $\sets_i$ for each buyer $i$, 
next we show how to find a profile of implementable pairs
$\lbr{\left(\virtualdist_i,\compensate_i\right)}_{i\in [n]}$
such that \(\rbr{\virtualdist_i,\compensate_i} \in \sets_i\) for all $i$ and the following term is maximized
\begin{align*}
\virtualwel(\virtualdists,\compensates) \triangleq \expect[\vals\sim \virtualdists]{\max_{i\in[n]}\val_i}+\sum_{i=1}^n \compensate_i.
\end{align*} 
By slightly overloading the notation, we also let 
$\virtualwel(\virtualdist,\compensate)\triangleq \expect[\val\sim \virtualdist]{\val}+\compensate$ when the input is a distribution and compensation pair for a single buyer.

We show that the optimal profile of implementable pairs can be computed by dynamic programming with an additive loss of $O(\varepsilon)$. 
At a high level, the DP algorithm progresses in $n$ stages. During the $i$-th stage, the DP memorizes all potential rounded distributions of the maximum virtual value among the first $i$ buyers. How do we proceed to stage $i+1$ using information of stage $i$? Suppose $\virtualdist_1 \in \feasibles$ represents a potential  distribution of the maximum virtual value among the first $i$ buyers, and $\virtualdist_2 \in \feasibles$ is a potential distribution of virtual value for buyer $i+1$.
Consider $\bar{\virtualdist} \triangleq \max\{\virtualdist_1, \virtualdist_2\}$ be the distribution of the maximum between two independent random variables $v_1$ and $v_2$, where $v_1\sim\virtualdist_1$ and $v_2\sim \virtualdist_2$. 
Next, we apply the procedure of $\discretize_Q$
to round down the probability masses of $\bar{\virtualdist}$ to the nearest multiples of $\frac{1}{nm}\varepsilon^4$. After the rounding procedure, a potential distribution is computed and stored for stage $i+1$. We show the rounding procedure only introduces negligible loss in revenue.
We don't need to discretize the virtual value space as $\bar{\virtualdist}$ is already supported on $\grid$.

Formally, we introduce the dynamic programming states $\dpf\rbr{i,\virtualdist,\compensate}$ as a Boolean variable for all buyer~$i$ and all feasible pairs $(\virtualdist, \compensate)\in\feasibles$, 
where the variable is True if there exists \( \lbr{\rbr{\virtualdist_1,\compensate_1},\cdots,\rbr{\virtualdist_i, \compensate_i}}\) such that
\begin{enumerate}
\item $\rbr{\virtualdist_j,\compensate_j} \in \sets_j$ for all $j \leq i$;
\item $\sum_{j=1}^i \compensate_j = \compensate$;
\item \sloppy if we let $\hat{\virtualdist}_1,\dots, \hat{\virtualdist}_i$ be a sequence of distributions such that $\hat{\virtualdist}_1 = \virtualdist_1$ and 
$\hat{\virtualdist}_j = \discretize_Q(\max(\hat{\virtualdist}_{j-1},\virtualdist_j))$ for any $2\leq j\leq i$, 
$\hat{\virtualdist}_i$ is equal to $\virtualdist$.
\end{enumerate}
Algorithm \ref{alg:dpf} shows how to compute $\dpf\rbr{i,\virtualdist,\compensate}$.

\begin{algorithm}[t]
\caption{Computation of Implementable Pairs}
\label{alg:dps}
\KwIn{Distribution \(\dist_i\) with support $\valspace_i$ and $\epsilon > 0$.}
\KwOut{Set of implementable pairs $\sets_i$.}

{Initialize every \(\dps\rbr{k',j, j', \virtualdist, \compensate}\) as $\left(\textsf{False}, {-\infty}, \cdot\right)$ for all $k'\leq k, j,j' \leq m_i+1, (\virtualdist, \compensate)\in\feasibles$. }

{\(\dps\rbr{0,m_i + 1, m_i + 1, \virtualdist_0, 0} \gets (\textsf{True}, \max\lbr{\val_i\in\valspace_i},\cdot)\) where $\virtualdist_0$ is the distribution that is deterministically 0.}

\For{$j \gets m_i$ \KwTo $1$ }{
    \For{$\rbr{k', j_l, j_r, \virtualdist,\compensate} \text{ s.t. } \dps_1\rbr{k',j_l, j_r, \virtualdist,\compensate} = \textsf{True} \text{ and } k' < k $}{
        \(\textsf{Prob} \gets \sum_{l = j}^{j_l - 1} \density_i(\val_{i, l});\,
        \textsf{ExpectedValue} \gets \rbr{\sum_{l = j}^{j_l-1} v_{i,l}\cdot \density_i(\val_{i, l})}/{\textsf{Prob}} \).\\
         $\textsf{VirtualValue} \gets \textsf{ExpectedValue} - \left(\frac{\sum_{l=j_l}^{j_r}  \val_{i,l}\cdot \density_i(\val_{i, l}) }{\sum_{l=j_l}^{j_r}   \density_i(\val_{i, l})}-\textsf{ExpectedValue} \right)
         \cdot \frac{\sum_{l=j_l}^{m_i}   \density_i(\val_{i, l})}{\textsf{Prob}}
         \cdot \indicate{j_r\neq m_i+1}$.\\
        \(\virtualdist'\gets \virtualdist\);
        \(\compensate' \gets \compensate\).\\
        \If{\(\textsf{VirtualValue} > \varepsilon^{-1}\)}{
            \(\compensate' \gets \compensate' + \textsf{VirtualValue}\cdot\textsf{Prob}\)
        } \Else{
            \text{Move a probability mass of } \textsf{Prob} \text{from} 0 \text{to} $\max\rbr{\textsf{VirtualValue}, 0}$ in \(\virtualdist'\).
        }
        {\(\rbr{\virtualdist',\compensate'} \gets \discretize\rbr{\virtualdist', \compensate'}\).\label{step:discretize_prob_each_value}\\}
        \If{$\textsf{VirtualValue} < \dps_2\rbr{k',j_l, j_r, \virtualdist,\compensate}$}
        {\(\dps\rbr{k' + 1, j, j_l - 1, \virtualdist', \compensate'} \gets \rbr{\textsf{True},\max\lbr{\textsf{VirtualValue},\dps_2\rbr{k' + 1, j, j_l - 1, \virtualdist', \compensate'}}, (k',j_l,j_r,G,c)}\).}
    }
}

\For{$\rbr{k',j,\virtualdist_i, \compensate_i} \text{ s.t. } \dps_1\rbr{k',1, j, \virtualdist_{i},\compensate_{i}} = \textsf{True}$\,}{
        Add \(\rbr{\virtualdist_i, \compensate_i}\) to \(\sets_i\).
    }
\end{algorithm}

\begin{algorithm}[H]
\caption{Computation of the Virtual Welfare}
\label{alg:dpf}
\KwIn{Sets of Implementable Pairs \(\lbr{\sets_i}_{i\in [n]}\).} 
\KwOut{The matrix 
\(\dpf\rbr{i, \virtualdist, \compensate}\).}
{Initialize \(\dpf\rbr{i, \virtualdist, \compensate}\) to be \textsf{False} for all $i\in[n]$ and $(\virtualdist, \compensate)\in\feasibles$. }

{\(\dpf\rbr{1, \virtualdist, \compensate} \gets \textsf{True}\)} 
for all $(\virtualdist, \compensate) \in \sets_1$. 

\For{$i \gets 2$ \KwTo $n$ }{
    \For{$\rbr{\virtualdist_{i-1},\compensate_{i-1}} \text{ s.t. } \dpf\rbr{i-1,\virtualdist_{i-1},\compensate_{i-1}} = \textsf{True}$ \,}{
        \For{\(\rbr{\virtualdist_{i}, \compensate_i} \in \sets_i \)} {
        $\dpf\rbr{i, \discretize_Q(\max\rbr{\virtualdist_{i-1}, \virtualdist_i}), \compensate_{i-1}+\compensate_{i}} \gets \textsf{True}$. \label{step:discrete_prob}
        }
    }
}

\end{algorithm}

With \(\dpf\rbr{n,\virtualdist,\compensate}\), the optimal signal structure can be found by identifying the feasible pair $(\virtualdist, \compensate)$ such that 
\(\dpf\rbr{n,\virtualdist,\compensate} = \textsf{True}\) and 
$\virtualwel(\virtualdist, \compensate)$ is maximized.
The full procedure of the PTAS is provided in Algorithm \ref{alg:ptas}.

\begin{algorithm}[H]
\caption{PTAS}
\label{alg:ptas}
\KwIn{Parameter $\epsilon > 0$ and value distributions $\dist_1,\dots,\dist_n$. }

\For{$i \gets 1$ \KwTo $n$ }{
    Compute $\sets_i$ for all buyer $i$ via Algorithm~\ref{alg:dps}. 
}

Compute $\dpf$ by running Algorithm~\ref{alg:dpf} with inputs \(\left\{\sets_i\right\}_{i\in [n]}.\)\\

Let \(\virtualdist, \compensate\) be the pair such that \(\dpf\rbr{n,\virtualdist,\compensate} = \textsf{True}\) and $\virtualwel(\virtualdist, \compensate)$ is maximized,
backtrack \(\dpf\rbr{n,\virtualdist,\compensate}\) to find the partitional signal structure $(\signals_i,\experiment_i)$ for each buyer $i$.

Compute the revenue optimal mechanism given signal structures $\lbr{(\signals_i,\experiment_i)}_{i\in[n]}$ according to \citet{myerson1981optimal}.
\end{algorithm}

\subsubsection{Proof of Revenue Loss}
\label{subsubsec:ptasdiscre}

We next proceed to show that such scheme introduces at most $O(\varepsilon)$ loss. Our proof consists of two parts. In the first part, we show that the discretization error from replacing a virtual value distribution with its corresponding discretized pair using Algorithm \ref{alg:dps} is at most $O\rbr{\varepsilon}$. 
In the second part, we show that the loss is at most $O\rbr{\frac{\epsilon}{n}}$ if we only round down the probability space. 
The main reason we analyze this discretization error separately is because we need to repeatedly apply the discretization in probability space in Step \ref{step:discrete_prob} of Algorithm \ref{alg:dpf}, while the discretization over virtual value space is only applied once for obtaining the set of implementable pairs.

We first show that the error from further discretization in Algorithm~\ref{alg:dps} is small. The proof of \Cref{lem:Algorithm2CorrectNess} can be found in \Cref{apx:revenue loss}.
\begin{lemma}\label{lem:discrete_algo2_loss}
Given any distribution profile $\virtualdists$ with support size $|\virtualdist_i|\leq m$ for all $i$ and
$\virtualwel(\virtualdists,0) \leq 1$, 
let $(\virtualdists',\compensates')=\discretize(\virtualdists)$ be the discretized pair of virtual value distributions and compensation terms. It holds that
\begin{align*}
\abs{\virtualwel(\virtualdists,0) - \virtualwel(\virtualdists',\compensates')} \leq 5\epsilon.
\end{align*}
\end{lemma}
The proof of \cref{lem:discrete_algo2_loss} proceeds in three steps. First, we decompose the virtual value distribution to store the realization of high virtual values into a separate compensation term. 
Essentially, the decomposition overestimates the true virtual welfare since for the high virtual values above the threshold of $\epsilon^{-1}$, the contribution to the virtual welfare is computed as the sum instead of the max. 
However, by Markov's inequality, the probability there exists a high virtual value among $n$ bidders is at most $\epsilon$ since the virtual welfare is at most $1$. 
Therefore, in such small probability events, the difference between the sum and the max is sufficiently small.
Secondly, to bound the error for discretization in the virtual value space, it is sufficient to couple the realizations over virtual value profiles given distributions before and after the discretization. 
Using this coupling argument, we show that the error in the virtual welfare is upper bounded by $O(\epsilon)$ without the additional multiplicative overhead of $n$. This is crucial as it limits the number of possible distributions to be polynomial in $n$ for every fixed $\varepsilon$.
Finally, the discretization in probability space is applied separately for each virtual value in the support of $\virtualdist_i$. 
The discretization error is small since the discretization grid is $\frac{\epsilon^4}{nm}$ and the virtual value distribution $\virtualdist$ has support size~$m$.
Note that the requirements on support size are naturally satisfied in our joint design problem since we are focusing on monotone partitional signal structures, which generates at most $m$ signals. 

Next we show that the error from Algorithm~\ref{alg:dpf} is small.
In particular, we show that the error from discretization over two distributions and then taking the max is at most $O(\frac{\epsilon}{n})$. See~\Cref{apx:revenue loss} for the proof of \Cref{lem:discrete_quant_loss}.

\begin{lemma}\label{lem:discrete_quant_loss}
Given two distributions $\virtualdist_1,\virtualdist_2$ with support in $\grid$, let $\virtualdist_i'=\discretize_Q(\virtualdist_i)$ for $i\in\{1,2\}$ be the distributions discretized in probability space. 
It holds that
\[\expect[\val\sim \max\rbr{\virtualdist_1,\virtualdist_2}]{\val} - \frac{2\epsilon}{n} \leq \expect[\val'\sim \max\rbr{\virtualdist_1',\virtualdist_2'}]{\val'}\leq \expect[\val\sim \max\rbr{\virtualdist_1,\virtualdist_2}]{\val}\]
\end{lemma}
This is because the discretization grid in probability space is {$\frac{1}{nm}\epsilon^4$}, while the maximum value in the support of $\grid$ is at most $\epsilon^{-1}$ and $|\grid| = O(\epsilon^{-2})$. Therefore, it is easy to verify that the expected virtual welfare loss from rounding down the probabilities can be bounded by $O\left(\frac{\epsilon}{n}\right)$.

\sloppy
\begin{proof}[Proof of Theorem~\ref{thm:ptas}]
First, it is easy to verify that the running time of  Algorithm~\ref{alg:ptas} is $\poly\rbr{\left(\frac{nm}{\epsilon}\right)^{f\rbr{\frac{1}{\varepsilon}}},n,m,k}$ since the cardinality of the feasible pairs is $O\rbr{\rbr{\frac{nm}{\epsilon}}^{f\left(\frac{1}{\varepsilon}\right)}}$ by \cref{lem:numberofdis}.

To simplify the exposition, we first introduce the following notation. 
For any profile of signal structures $\lbr{(\signals_i,\experiment_i)}_{i\in[n]}$, 
let $\rev\rbr{\lbr{(\signals_i,\experiment_i)}_{i\in[n]}}$ be the optimal revenue from mechanisms in which the signal structure coincides with $\lbr{(\signals_i,\experiment_i)}_{i\in[n]}$. 
Let $\virtualwel\rbr{\lbr{(\signals_i,\experiment_i)}_{i\in[n]}}$ be the virtual welfare given posterior value distribution generated by $\lbr{(\signals_i,\experiment_i)}_{i\in[n]}$. 
For any value distribution profile $\dists$, let $\lbr{(\signals^*_i,\experiment^*_i)}_{i\in[n]}$ be the profile of signal structures that maximizes the expected revenue. 
Since the induced posterior value distribution is regular \citep{bergemann2007information}, \cref{lem:rev_equiv} implies that 
\begin{align*}
\rev\rbr{\lbr{(\signals^*_i,\experiment^*_i)}_{i\in[n]}}
= \virtualwel\rbr{\lbr{(\signals^*_i,\experiment^*_i)}_{i\in[n]}}.
\end{align*}
Letting $(\virtualdists^*,\compensates^*)$ be the discretized distribution and compensation term pairs induced by the optimal signal structure $\lbr{(\signals^*_i,\experiment^*_i)}_{i\in[n]}$, 
the correctness of Algorithm \ref{alg:dps} illustrated in \cref{lem:Algorithm2CorrectNess} implies that $(\virtualdist^*_i,\compensate^*_i)\in\sets_i$ for all $i\in[n]$.
Moreover, \cref{lem:discrete_algo2_loss} implies that 
\begin{align*}
\abs{\virtualwel\rbr{\lbr{(\signals^*_i,\experiment^*_i)}_{i\in[n]}} - \virtualwel\rbr{\virtualdists^*,\compensates^*}} \leq 5\epsilon.
\end{align*}
since the optimal signal structure is monotone partitional and has cardinality at most $m$ for each buyer. 
Let $(\hat{\virtualdists},\hat{\compensates})$ be the discretized distribution and compensation term pairs computed by Algorithm \ref{alg:ptas}, and let
$\lbr{(\hat{\signals}_i,\hat{\experiment}_i)}_{i\in[n]}$ be the corresponding monotone paritional signal structure guaranteed by \Cref{lem:Algorithm2CorrectNess}. Similarly, we have 
\begin{align*}
\abs{\virtualwel\rbr{\lbr{(\hat{\signals}_i,\hat{\experiment}_i)}_{i\in[n]}} - \virtualwel\rbr{\hat{\virtualdists},\hat{\compensates}}} \leq 5\epsilon.
\end{align*}
Moreover, Algorithm \ref{alg:dps} ensures that the associated signal structures induce regular distributions, and hence 
\begin{align*}
\rev\rbr{\lbr{(\hat{\signals}_i,\hat{\experiment}_i)}_{i\in[n]}}
= \virtualwel\rbr{\lbr{(\hat{\signals}_i,\hat{\experiment}_i)}_{i\in[n]}}.
\end{align*}
Let $(\bar{\virtualdists},\bar{\compensates})$ be the discretized distribution and compensation term pairs that maximizes the virtual welfare, i.e., 
$(\bar{\virtualdists},\bar{\compensates}) = \argmax_{(\virtualdists,\compensates)\in\{\sets_i\}_{i\in[n]}}\virtualwel(\virtualdists,\compensates)$.
By definition, we have $\virtualwel(\bar{\virtualdists},\bar{\compensates}) \geq \virtualwel(\virtualdists^*,\compensates^*)$.
Note that in Algorithm \ref{alg:dpf}, we repeat the discretization in probability space Step \ref{step:discrete_prob} for $n$ at most times, and \cref{lem:discrete_quant_loss} shows that the error induced by each step is at most $\frac{2\epsilon}{n}$.
Therefore, the aggregated discretization error is at most $2\epsilon$, which further implies that 
\begin{align*}
\virtualwel\rbr{\hat{\virtualdists},\hat{\compensates}}
\geq \virtualwel(\bar{\virtualdists},\bar{\compensates}) -2\epsilon
\geq \virtualwel(\virtualdists^*,\compensates^*) -2\epsilon
\end{align*}
where the last inequality holds since $(\virtualdists^*,\compensates^*)$ is an implementable discretization pair. 
By combining all the inequalities above, we have 
\begin{align*}
\rev\rbr{\lbr{(\hat{\signals}_i,\hat{\experiment}_i)}_{i\in[n]}}
\geq \rev\rbr{\lbr{(\signals^*_i,\experiment^*_i)}_{i\in[n]}} - 12\epsilon
= \optrev - 12\epsilon.
\end{align*}

Finally, since \cref{cor:rev_approx_wel} implies that the optimal revenue is at least $O(1)$, the loss in revenue is also multiplicative in $O(\epsilon)$.
\end{proof}
    

%% file: content/conclusion.tex
\section{Conclusions and Discussions}
\label{sec:conclude}
In this paper, we show that the optimal joint design problem is NP-hard, and we provide a PTAS for computing the optimal solution. Our paper initiates a new research direction for understanding the computational complexity for jointly designing the signal and incentive structures. Numerous immediate questions emerge from our study, which we believe to be very intriguing. For instance, while our paper has concentrated on the single-item auction setting, it would be interesting to explore settings where the seller has multiple items for sale. Another worthwhile direction is investigating settings where the utility functions of the agents are interdependent. Finally, even within the model studied in this paper, there are compelling open questions, such as characterizing the maximum cardinalities of the optimal signal structures or showing whether FPTAS exists for computing the optimal solution.

%% file: appendix/nph.tex
\section{NP-hardness}
\label{apx:nph}

Before presenting the missing proofs in \Cref{sec:nph}, we first show that both the \probkname and the \probname problem is solvable in polynomial time when \(\max|V_i| = 2\).

\begin{theorem}
\label{thm:poly_support_2}
    Both the \probkname and the \probname problem can be solved in polynomial time when \(\max|V_i| = 2\).
\end{theorem}
\begin{proof}
    Without loss of generality, we can assume that \(k = 2\) for the \probkname problem since the support size is at most \(2\). Thus in this case, the solutions to the \probname problem and the \probkname are the same. In the following, we present the algorithm to compute the optimal solution.

\Cref{lem:partition} establishes that the optimal information structure for each buyer is required to be a monotone partitional information structure. Given that \(|V_i| \leq 2\), the range of possible information structures is limited to two scenarios: either both values in \(V_i\) are associated with the same signal, or each value is linked to a distinct signal.

Since we have at most \(2\) signals, it follows that the induced distributions over virtual values are automatically regular. Therefore, given the information structure for each buyer, the revenue of the optimal mechanism equals to the expected of the maximum virtual welfare:
\[\expect[\vals\sim \dists ]{\max_{i\in [n]} \virtual_i\rbr{\experiment_i\rbr{\val_i}}}.\]

Suppose there is at least one buyer who has only a single signal in the signal structure and let~\(i^*\) be the buyer who has the highest virtual value of that signal. Corresponding, define \(\virtual_{i^*}\) as the virtual value of this signal. It is clear that for any other signal with a virtual value lower than \(\virtual_{i^*}\), the contribution of this signal to the expected revenue is \(0\), as the maximum of the virtual value is always at least \(\virtual_{i^*}\). 

Next we show that it is without loss to consider the signal structure where $i^*$ is the only buyer with a single signal. 
For any other buyer \(i'\neq i^*\), we first assume that both values in \(V_{i'}\) are mapped to the same signal and compute the virtual value of the signal. 
First, if the virtual value of the signal is larger than \(\virtual_{i^*}\), it follows that each value in \(\valspace_{i'}\) must be linked to a distinct signal, as we assume that \(i^*\) is the one with the highest virtual value. If the virtual value of the signal is at most \(\virtual_{i^*}\), we can also map each value in \(\valspace_{i'}\) to different signals. This is because that mapping to the same signal has zero contribution to the optimal revenue. Consequently, assigning different signals to each value has non-negative impact on the revenue, which is never a worse choice.
Therefore, the optimal information structure is uniquely determined by fixing \(i^*\). By enumerating \(i^*\), this problem is solvable within polynomial time. 

In scenarios where \(i^*\) is non-existent, the signal structure for each buyer becomes predetermined, and the expected revenue in this case can be computed in polynomial time as well. Combining the two cases together, we complete the proof of \Cref{thm:poly_support_2}.
\end{proof}

\begin{proof}[Proof of \cref{lem:nph_support3}]
The proof is inspired by \cite{xiao2020complexity} and \cite{chen2014complexity}. The NP-hardness is derived through a reduction from the well-known NP-complete \textsc{Partition} problem \citep{garey1979computers}.

\textsc{Partition:} Given \(n\) positive integers \(c_1, c_2, \cdots, c_n\), is it possible to find a subset \(S\subseteq [n]\) such that \(\sum_{i\in S} c_i = \sum_{i\notin S} c_i\)?

{Given an instance of \textsc{Partition}  parameterized by \(c_1, c_2, \cdots, c_n\), we construct the hard instance as we described in Section~\ref{sec:nph}. Here we restate the process for the completeness of our proof. Given \(c_1 \geq c_2 \geq \cdots \geq c_n\), we correspondingly define  \(M = c_1\cdot 2^n\) and \(t_i = \half \sum_{\tl{j\in[n]}{j\neq i}} c_j / M\). For each \(i \in [n]\), define \(r_i\) and \(q_i\) as the unique, non-negative solutions to the following system of equations:}


\begin{align}
\label{eq:Support3NPH1}r_i + q_i &= c_i / M.\\ 
2q_i + r_i t_i &= q_i + r_i \label{eq:Support3NPH2}
\end{align}

In particular, we have
\begin{align*}
r_i = \frac{c_i}{M(2-t_i)} 
\quad\text{and}\quad
q_i = \frac{c_i(1-t_i)}{M(2-t_i)}.
\end{align*}

The hard instance for the \probname problem is constructed as follows comprising \(n + 1\) buyers. Specifically, for the \(i\)-th buyer among the first \(n\) buyers, the value distribution \(F_i\) is defined as follows:
\begin{align*}
    v_i \sim F_i, v_i = \begin{cases}
        3 & \text{w.p. }  q_i\\
        t_i + 1 & \text{w.p. } r_i\\
        0 & \text{w.p. } 1 - q_i - r_i
    \end{cases}
\end{align*}

The value distribution of the \(\rbr{n + 1}\)-th bidder is deterministically \(1\). According to Lemma~\ref{lem:partition}, there always exists an optimal signal structure that is also monotone and partitional. For the \(i\)-th buyer where \(i \leq n\), there are four distinct signal structures available for selection: \(\rbr{\lbr{3}, \lbr{t_i + 1},  \lbr{0}}\), \(\rbr{\lbr{3,t_1 + 1},  \lbr{0}}\), \(\rbr{\lbr{3}, \lbr{t_i + 1,0}}\) or \(\rbr{\lbr{3,t_i + 1,0}}\). First, it is important to recognize that there always exists an optimal signal structure which does not include  \(\rbr{\lbr{3,t_i + 1,0}}\). 
This is because pooling \(0\) is dominated by not pooling it, and switching from \(\rbr{\lbr{3,t_i + 1,0}}\) to \(\rbr{\lbr{3,t_i + 1},  \lbr{0}}\) always results in a weak increase in revenue.\footnote{This is because under signal structure $\rbr{\lbr{3,t_i + 1},  \lbr{0}}$, the seller can mimic the allocation rule of the mechanism under signal structure $\rbr{\lbr{3,t_i + 1,0}}$ for all buyers $j\neq i$ or buyer $i$ with value in $\lbr{3,t_i + 1}$, and not allocate the item to buyer $i$ if his value is~$0$. The expected virtual welfare remains unchanged 
given this mechanism and hence the expected revenue remains unchanged as well due to the revenue equivalence result in \citet{myerson1981optimal}. The optimal revenue given signal structure $\rbr{\lbr{3,t_i + 1},  \lbr{0}}$ can only be weakly larger.} 
Now let us consider the signal structure  \(\rbr{\lbr{3}, \lbr{t_i + 1,0}}\). It can be readily shown that the virtual value of the signal \(\lbrace t_i + 1, 0 \rbrace\) is given by
\begin{align*}
   \virtual\rbr{{\lbr{t_i+1, 0}}} & =  \frac{(1 + t_i) r_i}{(1 - q_i)} - \rbr{3 - \frac{(1+t_i)r_i}{(1 - q_i)}} \cdot \frac{q_i}{1 - q_i}\\
    & = \frac{3q_i^2 - 3q_i + r_i + r_i t_i}{(1 - q_i)^2} < 1.
\end{align*} where the last inequality holds since \(r_i, q_i, t_i \ll 1\). However, it is important to note that the buyer \(n + 1\) always has a virtual value of \(1\). As per the characterization of the optimal mechanism provided by \cite{myerson1981optimal}, it follows that \(x_i \left(\lbrace t_i + 1, 0 \rbrace\right) = 0\). This is because the optimal mechanism consistently favors the \(\rbr{n + 1}\)-th buyer, thereby resulting in a zero probability of allocating to the type \(\lbrace t_i + 1, 0 \rbrace\). Therefore, switching from \(\rbr{\lbr{3}, \lbr{t_i + 1,0}}\) to \(\rbr{\lbr{3}, \lbr{t_i + 1} ,\lbr{0}}\) always weakly increases the revenue, as we keep the virtual value of the signal \(\lbr{3}\) unchanged, and only splits the signal that has \(0\) probability to be allocated. Thus, we have demonstrated that there always exists an optimal signal structure which only contains two possible signal structures for the first \(n\) buyers: either  \(\rbr{\lbr{3}, \lbr{t_i + 1} ,\lbr{0}}\) or \(\rbr{\lbr{3, t_i + 1} \lbr{0}}\).

Given that the first \(n\) buyers have only two signal structure options, let us define \(S\) as the set of buyers opting for the signal structure \(\left(\lbrace 3, t_i + 1 \rbrace, \lbrace 0 \rbrace\right)\), and \(T\) as the set of those selecting \(\left(\lbrace 3 \rbrace, \lbrace t_i + 1 \rbrace, \lbrace 0 \rbrace\right)\). It is clear that for any buyer \(i \in S\), the virtual value associated with the signal \(\lbr{3, t_i + 1}\) is calculated as:
\begin{align*}
    \virtual\rbr{\lbr{3,t_i + 1}} = \frac{3q_i + r_i + r_it_i}{q_i + r_i} = 2.
\end{align*}where the last equality follows from \Cref{eq:Support3NPH2}.

For each buyer  $i\in T$, the virtual value for the signal \(\lbr{t_i+1}\) 
is given by:
\begin{align*}
    \virtual\rbr{\lbr{t_i+1}} = (t_i+1) - (3 - (t_i + 1)) \frac{q_i}{r_i} = (t_i+1) - (3 - (t_i + 1)) (1-t_i)< 1
\end{align*}
since $t_i<\frac{1}{2}$.

Furthermore, the virtual value of the signal \(\lbrace 0 \rbrace\) is always negative in both signal structures. Consequently, the optimal mechanism under this information structure is straightforward: The mechanism initially offers a price of \(3\) to the buyers in \(T\) in any order. If the item remains unsold, it then proposes a price of \(2\) to the buyers in \(S\). Should the item still be unsold, it will finally be sold to the \(n + 1\)-th buyer at a price of \(1\). Now, notice that both \(r_i\) and \(q_i\) are \(O\rbr{\frac{1}{M}}\). We approximate the revenue of the optimal mechanism under the sets \(S\) and \(T\) by expanding the contribution up to the second terms and ignoring the third order terms, i.e., terms of order \(O\rbr{\epsilon}\) where \(\epsilon = {n^3}/{M^3}\):
\begin{align}\label{eq:revenueST}
\begin{split}
&\rev(S, T) = 3\cdot \rbr{1 - \prod_{i\in T} \rbr{1 - q_i}} + 2\cdot \prod_{i\in T}\rbr{1 - q_i}\cdot \rbr{1 - \prod_{i\in S}\rbr{1 - q_i - r_i}} \\
&\qquad\qquad\qquad\qquad+ \rbr{1 - \prod_{i\in T}\rbr{1 - q_i}\prod_{i\in S}\rbr{1 - q_i - r_i}}\\
& = 1 + \rbr{2\cdot\sum_{i\in T} q_i + \sum_{i\in S}\rbr{q_i + r_i}} - \rbr{2\sum_{\tl{i,j\in T}{i< j}} q_i q_j + \sum_{\tl{i\in T}{j\in S}}q_i\rbr{q_j + r_j}+\sum_{\tl{i,j\in S}{i < j}}\rbr{q_i + r_i}\rbr{q_j + r_j}} \pm O\rbr{\epsilon}.
\end{split}
\end{align}

For the constant and the first order term, it holds that
\begin{align}\label{eq:firstorder}
\begin{split}    
    1 + 2\cdot \sum_{i\in T} q_i + \sum_{i\in S} (q_i + r_i) &= 1 + 2\sum_{i\in [n]} q_i + \sum_{i \in S} r_i t_i\\
    & = 1  + 2\sum_{i\in [n]} q_i + \half \sum_{i \in {S}} r_i \sum_{\tl{j\in [n]}{j\neq i}}  c_j / M\\
    & = 1  + 2\sum_{i\in [n]} q_i + \half \sum_{i \in {S}} r_i \sum_{\tl{j\in [n]}{j\neq i}} \rbr{r_j + q_j},
\end{split}
\end{align} where the first equation is derived from Equation~\eqref{eq:Support3NPH2} and the second equation expands \(t_j\) according to its definition that \(t_i = \sum_{j \neq i} c_j/M\). The last equation is from \Cref{eq:Support3NPH1}. Notice that \Cref{eq:Support3NPH2} suggests that 
\begin{align*}
    r_i &= q_i + r_i t_i  = q_i + O\rbr{\frac{n}{M^2}}.
\end{align*}

This equation implies that 
\begin{align}\label{eq:secondorder1}
\begin{split}
\sum_{i\in T}\sum_{j\in S} q_i \rbr{q_j + r_j} & = \sum_{i\in T}\sum_{j\in S} q_i \rbr{2q_j + O\rbr{\frac{n}{M^2}}}\\
& = 2 \sum_{i\in T}\sum_{j\in S} q_i q_j + O(\epsilon).
\end{split}
\end{align}

Similarly, it can be shwon that
\begin{align}\label{eq:secondorder2}
\begin{split}
\sum_{\tl{i,j\in S}{i < j}} \rbr{q_i + r_i} \rbr{q_j + r_j} & = \sum_{\tl{i,j\in S}{i < j}} q_i \rbr{q_j + r_j} + \sum_{\tl{i,j\in S}{i < j}} r_i \rbr{q_j + r_j}\\
& = \sum_{\tl{i,j\in S}{i < j}} q_i \rbr{2q_j + O\rbr{\frac{n}{M^2}}} + \sum_{\tl{i,j\in S}{i < j}} r_i \rbr{q_j + r_j}\\
& =\sum_{\tl{i,j\in S}{i < j}} 2 q_i q_j + \sum_{\tl{i,j\in S}{i < j}} r_i \rbr{q_j + r_j} + O(\epsilon).
\end{split}
\end{align}


By substituting Equation~\eqref{eq:firstorder}, \eqref{eq:secondorder1} and \eqref{eq:secondorder2} into the right hand side of Equation~\eqref{eq:revenueST}, it follows that 
\begin{align}\label{eq:revenuesimplify}
\rev\rbr{S,T} &=  1  + 2\sum_{i\in [n]} q_i + \half \sum_{i \in {S}} r_i \sum_{\tl{j\in [n]}{j\neq i}} \rbr{r_j + q_j} \nonumber\\
&\qquad\qquad- \rbr{ \sum_{\tl{i,j\in T}{i< j}}2 q_i q_j + 2 \sum_{i\in T}\sum_{j\in S} q_i q_j  + \sum_{\tl{i,j\in S}{i < j}} 2 q_i q_j + \sum_{\tl{i,j\in S}{i < j}} r_i \rbr{q_j + r_j}} \pm O\rbr{\epsilon}\nonumber\\
& = 1 + 2\sum_{i\in [n]} q_i - 2 \sum_{\tl{i,j\in [n]}{i< j}} q_i q_j   + \half \sum_{i\in S}r_i\sum_{\tl{j\in [n]}{j\neq i}}\rbr{r_j + q_j} - \sum_{\tl{i,j\in S}{i < j}} r_i \rbr{r_j + q_j}\pm O(\epsilon)\nonumber\\
& = 1 + 2\sum_{i\in [n]} q_i - 2\sum_{\tl{i,j\in [n]}{i< j}} q_i q_j   + \half \sum_{i\in S}r_i\sum_{j\in T}\rbr{r_j + q_j} \pm O(\epsilon).
\end{align} 
In the second equality, we employ the following equation:
\[\sum_{\tl{i,j\in T}{i< j}} q_i q_j + \sum_{\tl{i,j\in S}{i< j}} q_i q_j + \sum_{i\in S, j\in T} q_i q_j = \sum_{\tl{i,j\in [n]}{i<j}} q_i q_j\]
In the third equality, we use the fact that \[\half \sum_{i\in S} r_i\sum_{\tl{j\in S} {j\neq i}}(r_j+q_j)=\sum_{\tl{i,j \in S}{i\neq j}}r_ir_j+O(\varepsilon)\] and \[\sum_{\tl{i,j \in S}{i< j}}r_i(q_j+r_j)=2\sum_{\tl{i,j \in S}{i< j}}r_ir_j+O(\varepsilon)=\sum_{\tl{i,j \in S}{i\neq j}}r_ir_j+O(\varepsilon).\]
It is clear that the expression $1 + 2\sum_{i\in [n]} q_i - 2\sum_{\substack{i<j\\i, j\in [n]}} q_i q_j$, denoted as \(L\), remains constant and is independent of the specific choices of \(S\) and \(T\). Regarding the only remaining term, adding \(r_i\) to both sides of \Cref{eq:Support3NPH2} and subsequently rearranging yields:
\begin{align*}
    r_i &= \half \rbr{r_i + q_i} + \half r_i t_i\\
        & = \half \rbr{r_i + q_i} \pm O(n^2 / M^2).
\end{align*}

Combining the equation above with \Cref{eq:revenuesimplify} and \Cref{eq:Support3NPH1}, it follows that
\begin{align*}
    \rev(S,T) = L + \frac{1}{4M^2}\rbr{\sum_{i\in S} c_i}\rbr{\sum_{j\in T} c_j} \pm O(\epsilon).
\end{align*}

Define \(H\) as \(\sum_{i\in [n]} c_i / 2\). When there exists an equal-sum partition \(S, T\), the revenue generated by the corresponding sets \(S, T\) is \(L + \frac{H^2}{4M^2} \pm O\left(\epsilon\right)\). Conversely, in the absence of such a partition, \Cref{lem:partition} asserts the existence of an optimal monotone partitional signal structure. Consequently, the maximum achievable revenue is capped at \(L + \frac{H^2 - 1}{4M^2} \pm O\left(\epsilon\right)\). Given that \(\epsilon = \frac{n^3}{M^3} \) is significantly smaller than \(\frac{1}{M^2}\), deciding whether the optimal revenue exceeds the threshold \(t = L + \frac{H^2 - 0.5}{4M^2}\) 
solves the \textsc{Partition} problem. 
This concludes our proof.
%
\end{proof}

\begin{proof}[Proof of \cref{lem:partition}]

First, \cite{bergemann2007information} have established that there exists an monotone optimal signal structure for both the \probname and the \probkname problem. A signal structure, denoted as \(\lbr{ \rbr{\signals_i, \experiment_i} }_{i\in [n]} \), is monotone, if and only if the following condition holds for all buyer~\(i\): For any signal \(\signal_i \in \signals_i\), and two values \(\val_a < \val_b \in \valspace_i\) such that \(\prob{\experiment_i\rbr{\val_a} = \signal_i } > 0, \prob{\experiment_i\rbr{\val_b} = \signal_i } > 0\), it holds that \(\prob{\experiment_i\rbr{\val_c} = \signal_i }  = 1\) for any $\val_c \in \valspace_i$ satisfying \(\val_c \in \rbr{\val_a, \val_b}\). 


We start by establishing that for any given \(k\), the \probkname problem possesses an optimal signal structure that is both monotone and partitional. A direct corollary of this monotonicity suggests that for any buyer \(i\) with a valuation support of at most \(m\) elements, that is, \(|\valspace_i| \leq m\), the optimal signal structure entails no more than \(2m\) distinct signals. Consequently, the solution to the \probname problem aligns identically with the solution to the \textsc{Optimal 2m-Signal} problem. Thus, the validity of this statement for the \probkname problem inherently affirms its correctness to the \probname problem as well.


We prove by contradiction. By \citet{bergemann2007information}, there is an optimal BIC-IIR mechanism $\rbr{\lbr{(\signals_i, \experiment_i)}_{i\in[n]},\lbr{\alloc_i}_{i\in[n]},\lbr{\pay_i}_{i\in[n]}}$ in which the signal structure \(\lbr{(\signals_i, \experiment_i)}_{i\in[n]}\) is monotone, and we assume that it is not a partition. 
Given the monotonicity of the signal structure, the signals for each buyer~\( i \) can be arranged in ascending order, denoted as \(\signal_{i, 1}, \signal_{i, 2}, \cdots, \signal_{i, |\signals_i|}\). Specifically, for any pair of signals \(\signal_{i,a}, \signal_{i, b} \in \signals_i\) where \(a < b\), and any two values \(\val_{i,a}, \val_{i,b} \in \valspace_i\) such that $\prob{\experiment_i\rbr{\val_{i,a}} = \signal_{i,a} } > 0$ and $\prob{\experiment_i\rbr{\val_{i, b}} = \signal_{i, b} } > 0$, it always holds that \(v_{i,a} \leq v_{i, b}\). 
As the signals are arranged in ascending order, the corresponding expected values follow the same order. Moreover, since the mechanism is BIC, \citet{myerson1981optimal} shows that the interim allocation probability of each signal is also monotone. That is, 
\begin{align*}
    \expect[\experiment, \dists]{\alloc_i\rbr{\signal_i, \signal_{-i}}\given \signal_i = \signal_{i, 1}} \leq  \expect[\experiment, \dists]{\alloc_i\rbr{\signal_i, \signal_{-i}}\given \signal_i = \signal_{i, 2}} \leq \cdots \leq  \expect[\experiment, \dists]{\alloc_i\rbr{\signal_i, \signal_{-i}}\given \signal_i = \signal_{i, |\signals_i|}}
\end{align*}

The first observation is that if there are two adjacent signals \(\signal_{i,a}, \signal_{i,a+1} \in \signals_i\) with identical interim allocation probabilities in the optimal mechanism, i.e., 
\begin{align*}
\expect[\experiment, \dists]{\alloc_i\rbr{\signal_i, \signal_{-i}}\given \signal_i = \signal_{i, a}} 
= \expect[\experiment, \dists]{\alloc_i\rbr{\signal_i, \signal_{-i}}\given \signal_i = \signal_{i, a + 1}} 
\end{align*}
their interim payments must also be the same since otherwise the mechanism is not BIC \citep{myerson1981optimal}. 
In this case, we can define a new signal structure \(\rbr{\signals_i', \experiment_i'}\) which combines the signals \(\signal_{i, a}\) and \(\signal_{i, a + 1}\) for each buyer $i$. Specifically, \(\signals_i' = \lbr{\signal_{i, 1},\cdots, \signal_{i, a}, \signal_{i, a + 2}, \cdots, \signal_{i, |\signals_i|}}\), and \(\prob{\experiment_i'\rbr{v_i} = \signal_{i,j}} = \prob{\experiment_i\rbr{v_i} = \signal_{i,j}}\) holds for all \(\val_i \in \valspace_i\) and \(j \not= a, a + 1 \). For the signal \(\signal_{i, a}\), we define that \[\prob{\experiment_i'\rbr{v_i} = \signal_{i,a}} = \prob{\experiment_i\rbr{v_i} = \signal_{i,a}} + \prob{\experiment_i\rbr{v_i} = \signal_{i, a + 1}}\] for all \(\val_i \in \valspace_i\). 
Moreover, the interim allocation and payment rules remain unchanged given the new signal structure. It is easy to verify that the alternative mechanism is still BIC-IIR with respect to the new signal structure \(\lbr{\rbr{\signals_i', \experiment_i'}}_{i\in[n]}\) and the expected revenue of the mechanism remains the same. 
Therefore, combining signals with the same interim allocation probability in the optimal mechanism leads to a weak increase in revenue. Thus, in the following, we can assume that \( \expect[\signal, \dists]{\alloc_i\rbr{\signal_i, \signal_{-i}}\given \signal_i = \signal_{i, j}} \) is strictly increasing with respect to  \(j\).

Before initiating our proof, it is necessary to introduce a new family of mechanisms induced by a permutation \(\per\) over a selected subset of all signals \(\lbrace \signal_{1, 1}, \cdots, \signal_{n, |\signals_n|} \rbrace\). The mechanism operates as follows: It sequentially iterates through all signals in the order specified by the permutation \(\per\), and allocates the item to the buyer who has the first occurrence of a signal in \(\per\). If none of the signals in the permutation \(\per\) appear, the mechanism simply does not allocate the item. Once the allocation rule is determined, the mechanism determines the payment according to the characterization by \cite{myerson1981optimal}.

\begin{algorithm}[H]
Each bidder $i$ submits their signal $\signal_i \in \signals_i$.\\
\For{\(j \gets 1\) \KwTo \(|\per| \)}{
    \(\signal'_{i, a} \gets \per_j \).\\
\tcc{Here assume that \(\per_j\) is a signal \(\signal'_{i, a} \) from buyer \(i\).}
    \If{\(\signal_i = \signal'_{i, a}\)}{
        Allocate the item to buyer \(i\). \\
        Calculate the payment of buyer \(i\) according to the result by \cite{myerson1981optimal}.\\
        Exit. \\
    }
}
Does not allocate the item to any buyer.
\caption{Mechanism Induced by a Permutation \(\per\)}\label{alg:permutation}
\end{algorithm}

Given a permutation \(\per\), we define \(\lbr{\alloc^{\per}_i\rbr{\signal_i, \signal_{-i}}}_{i\in [n]}, \lbr{\pay^{\per}_i\rbr{\signal_i, \signal_{-i}}}_{i\in [n]}\) as the allocation and payment rules, respectively, of the mechanism induced by \(\per\). We may abuse the notation and also use \(\per\) to refer the mechanism induced by a permutation \(\per\). \cite{bergemann2007information} have established that the optimal mechanism corresponding to the optimal signal structure can also be characterized as a mechanism induced by a permutation. A mechanism induced by a permutation \(\per\) is the BIC-IIR if and only if, for each buyer \(i\), the signals included in \(\per\) constitute a suffix of the complete set of their possible signals \(\lbr{\signal_{i, 1}, \signal_{i, 2}, \cdots, \signal_{i, |\signals_i|}}\). Furthermore, these included signals must be arranged within \(\per\) in descending order. 

Considering a specific signal \(\signal_{i, a} \in \per\) from buyer \(i\), the interim probability of allocating the signal \(\signal_{i, a}\) corresponds precisely to the probability that \(\signal_{i, a}\) is the first to appear in the permutation sequence:
\begin{align}
\label{eq:allocationprob}
\expect[\experiment,\dists]{\alloc_i^{\per}\rbr{\signal_i, \signal_{-i}}\given \signal_i = \signal_{i, a}} = \prod_{j \neq i}\rbr{1 - \sum_{\tl{\signal_{j}\in \signals_j}{\signal_j <_{\per} \signal_{i, a}} } \prob{\experiment_j\rbr{\val_j} = \signal_j}}
\end{align}

\cite{bergemann2007information} demonstrate that, for any mechanism induced by a permutation \(\per\), its revenue is given by:
\begin{align}
\label{eq:revenue}
\begin{split}
\sum_{i=1}^n \expect[\experiment,\dists]{\pay_i^{\per}\rbr{\signal_i, \signal_{-i}}} = &\sum_{i = 1}^n \sum_{a = 1}^{|\signals_i|} \rbr{\expect[\experiment,\dists]{\alloc_i^{\per}\rbr{\signal_i, \signal_{-i}}\given \signal_i = \signal_{i, a}} - \expect[\experiment,\dists]{\alloc_i^{\per}\rbr{\signal_i, \signal_{-i}}\given \signal_i = \signal_{i, a - 1}}} \cdot \\
& \quad \quad \expect{\val_i \given \experiment_i\rbr{\val_i} = \signal_{i, a}} \cdot \prob{\experiment_i(\val_i) \geq \signal_{i, a}}.
\end{split}
\end{align} Here, the notation \(\signal_{i, a} \geq \signal_{i, b}\) is used to denote that \(a \geq b\).

We are now ready to prove the lemma. Since the optimal monotone signal structure is not partitional, there exists a pair of adjacent signals \(\signal_{i,a}, \signal_{i,a+1} \in \signals_i\) and a value \(v_a \in \valspace_i\) such that both \(\prob{\experiment_i\rbr{v_a} = \signal_{i,a}} \) and \(\prob{\experiment_i\rbr{v_a} = \signal_{i,a + 1}} \) are strictly positive. As the signal structure is monotone, it is established that 
\begin{align*}
    \expect{\val_i \given \experiment_i\rbr{\val_i} = \signal_{i, a}} &\leq \val_a\\
    \expect{\val_i \given \experiment_i\rbr{\val_i} = \signal_{i, a + 1}} &\geq \val_a
\end{align*}
Furthermore, it is not possible for both inequalities to be equal simultaneously. This is from the observation that should both become equalities, merging the two signals would lead to unchanged revenue.

In the subsequent analysis, we introduce two alternative signal structures with the aim of demonstrating that better one of these two yields strictly higher revenue. Let \(\varepsilon>0\) be a sufficiently small variable. The first alternative signal structure, denoted as \(\lbrace (\signals_i', \experiment_i') \rbrace_{i\in[n]}\), involves transferring a probability mass of \(\varepsilon\) for \(\val_a\) from signal \(\signal_{i, a}\) to \(\signal_{i, a + 1}\). Formally, \((\signals_i', \experiment'_j)\) is identical to \((\signals_i, \experiment_j)\) for all \(j\in [n]\), except the case that 
\begin{align*}
     \prob{\experiment_i'(\val_a) = \signal_{i, a}} &= \prob{\experiment_i(\val_a) = \signal_{i, a}} - \varepsilon\\
     \prob{\experiment_i'(\val_a) = \signal_{i, a + 1}} &= \prob{\experiment_i(\val_a) = \signal_{i, a + 1}} + \varepsilon.
\end{align*}

For the second signal structure \(\lbrace (\signals_i'', \experiment_i'') \rbrace_{i\in[n]}\), it just operates in the opposite direction and moves a probability mass of \(\varepsilon\) for \(\val_a\) from signal \(\signal_{i, a + 1}\) to \(\signal_{i, a}\). Specifically, \((\signals_i'', \experiment''_j)\) remains identical to \((\signals_i, \experiment_j)\) for all \(j\in [n]\), with the exception that:
\begin{align*}
     \prob{\experiment_i''(\val_a) = \signal_{i, a}} &= \prob{\experiment_i(\val_a) = \signal_{i, a}}  + \varepsilon\\
     \prob{\experiment_i''(\val_a) = \signal_{i, a + 1}} &= \prob{\experiment_i(\val_a) = \signal_{i, a + 1}} -  \varepsilon.
\end{align*}

Let $\per$ be the permutation corresponding to the optimal mechanism with respect to the original signal structure \(\lbr{\signals_i, \experiment_i)}_{i\in[n]}\). Instead of focusing directly on the optimal revenue associated with the signal structures \(\lbrace (\signals_i', \experiment_i') \rbrace_{i\in[n]}\) and \(\lbrace (\signals_i'', \experiment_i'') \rbrace_{i\in[n]}\), we examine the revenue generated by the mechanism induced by the permutation \(\per\). For simplicity, define \begin{align*}
    w_a = \expect{\val_i \given \experiment_i\rbr{\val_i} = \signal_{i, a}}\quad  & q_a = \prob{ \experiment_i\rbr{\val_i} = \signal_{i, a}}\\
    w_{a + 1} = \expect{\val_i \given \experiment_i\rbr{\val_i} = \signal_{i, a + 1}}\quad  & q_{a + 1} = \prob{ \experiment_i\rbr{\val_i} = \signal_{i, {a + 1}}}
\end{align*} We define similar variables for \(\lbrace (\signals_i', \experiment_i') \rbrace_{i\in[n]}\) and \(\lbrace (\signals_i'', \experiment_i'') \rbrace_{i\in[n]}\) in the same way. According to our definition, it directly follows that 
\begin{align*}
    w_a' = \frac{w_a \cdot q_a - \varepsilon \val_a}{q_a - \varepsilon} \quad &   w_{a + 1}' = \frac{w_{a + 1} \cdot q_{a + 1}  + \varepsilon \val_a}{q_{a + 1} + \varepsilon} \\
    w_a'' = \frac{w_a \cdot q_a + \varepsilon \val_a}{q_a + \varepsilon} \quad &   w_{a + 1}'' = \frac{w_{a + 1} \cdot q_{a + 1}  - \varepsilon \val_a}{q_{a + 1} - \varepsilon} 
\end{align*}
Define \(\Delta'\) and \(\Delta''\) as the differences in revenue between the mechanisms corresponding to the signal structures \(\lbrace (\signals_i', \experiment_i') \rbrace_{i\in[n]}\) and \(\lbrace (\signals_i'', \experiment_i'') \rbrace_{i\in[n]}\), when induced by \(\per\), and the optimal revenue associated with the original signal structure \(\lbrace (\signals_i, \experiment_i) \rbrace_{i\in[n]}\). The variations in revenue can be attributed to three distinct components: (1) The contribution of \(\signal_{i, a}\); (2) The contribution of \(\signal_{i, a + 1}\); (3) The contribution from signals \(\signal_{j, b}\) of other buyers \(j \neq i\), wherein either \(\expect[\experiment,\dists]{\alloc_j^{\per}(\signal_j, \signal_{-j}) | \signal_j = \signal_{j, b}}\) or \(\expect[\experiment,\dists]{\alloc_j^{\per}(\signal_j, \signal_{-j}) | \signal_j = \signal_{j, {b - 1}}}\) experiences a change. 

We commence by computing the contribution from component (1), denoting the respective differences by \(\Delta_{(1)}'\) and \(\Delta_{(1)}{''}\). Notice that 
\begin{align*}
    \Delta_{(1)}' & =  \rbr{\expect[\experiment',\dists]{{\alloc'_i}^{\per}\rbr{\signal'_i, \signal'_{-i}}\given \signal'_i = \signal'_{i, a}} - \expect[\experiment',\dists]{{\alloc'_i}^{\per}\rbr{\signal'_i, \signal'_{-i}}\given \signal'_i = \signal'_{i, a - 1}}} \cdot w_a' \cdot \prob{\experiment'_i(\val_i) \geq \signal'_{i, a}} \\
    & \quad  - \rbr{\expect[\experiment,\dists]{\alloc_i^{\per}\rbr{\signal_i, \signal_{-i}}\given \signal_i = \signal_{i, a}} - \expect[\experiment,\dists]{\alloc_i^{\per}\rbr{\signal_i, \signal_{-i}}\given \signal_i = \signal_{i, a - 1}}} \cdot w_a \cdot \prob{\experiment_i(\val_i) \geq \signal_{i, a}} \\
    & = \rbr{\expect[\experiment,\dists]{\alloc_i^{\per}\rbr{\signal_i, \signal_{-i}}\given \signal_i = \signal_{i, a}} - \expect[\experiment,\dists]{\alloc_i^{\per}\rbr{\signal_i, \signal_{-i}}\given \signal_i = \signal_{i, a - 1}}}\cdot \prob{\experiment_i(\val_i) \geq \signal_{i, a}} \\
    & \quad \quad \cdot \frac{\epsilon \rbr{w_a - \val_a}}{q_a - \epsilon}
\end{align*} where in the second equality we uses the fact that \(\expect[\experiment,\dists]{\alloc_i^{\per}\rbr{\signal_i, \signal_{-i}}\given \signal_i = \signal_{i, a}}\), \(\expect[\experiment,\dists]{\alloc_i^{\per}\rbr{\signal_i, \signal_{-i}}\given \signal_i = \signal_{i, a - 1}}\) and \(\prob{\experiment_i(\val_i) \geq \signal_{i, a}}\) remains unchanged following the transfer of probability. Similarly, 
\begin{align*}
    \Delta_{(1)}'' & =  \rbr{\expect[\experiment',\dists]{{\alloc''_i}^{\per}\rbr{\signal''_i, \signal''_{-i}}\given \signal''_i = \signal''_{i, a}} - \expect[\experiment',\dists]{{\alloc''_i}^{\per}\rbr{\signal''_i, \signal''_{-i}}\given \signal''_i = \signal''_{i, a - 1}}} \cdot w_a' \cdot \prob{\experiment''_i(\val_i) \geq \signal''_{i, a}} \\
    & \quad  - \rbr{\expect[\experiment,\dists]{\alloc_i^{\per}\rbr{\signal_i, \signal_{-i}}\given \signal_i = \signal_{i, a}} - \expect[\experiment,\dists]{\alloc_i^{\per}\rbr{\signal_i, \signal_{-i}}\given \signal_i = \signal_{i, a - 1}}} \cdot w_a \cdot \prob{\experiment_i(\val_i) \geq \signal_{i, a}} \\
    & = \rbr{\expect[\experiment,\dists]{\alloc_i^{\per}\rbr{\signal_i, \signal_{-i}}\given \signal_i = \signal_{i, a}} - \expect[\experiment,\dists]{\alloc_i^{\per}\rbr{\signal_i, \signal_{-i}}\given \signal_i = \signal_{i, a - 1}}}\cdot \prob{\experiment_i(\val_i) \geq \signal_{i, a}} \\
    & \quad \quad \cdot \frac{\epsilon \rbr{\val_a - w_a}}{q_a + \epsilon}
\end{align*}

Combining the two equations above implies that \begin{align*}
    &\Delta_{(1)}' + \Delta_{(1)}'' \\
    = &\rbr{\expect[\experiment,\dists]{\alloc_i^{\per}\rbr{\signal_i, \signal_{-i}}\given \signal_i = \signal_{i, a}} - \expect[\experiment,\dists]{\alloc_i^{\per}\rbr{\signal_i, \signal_{-i}}\given \signal_i = \signal_{i, a - 1}}} \prob{\experiment_i(\val_i) \geq \signal_{i, a}} \frac{2\varepsilon^2 \rbr{\val_a - w_a}}{q_a^2 - \varepsilon^2} \geq 0
\end{align*}

Now let us compute the component (2). It is true that 
\begin{align*}
    & \Delta_{(2)}' + \Delta_{(2)}'' \\
    & =  \rbr{\expect[\experiment',\dists]{{\alloc'_i}^{\per}\rbr{\signal'_i, \signal'_{-i}}\given \signal'_i = \signal'_{i, a + 1}} - \expect[\experiment',\dists]{{\alloc'_i}^{\per}\rbr{\signal'_i, \signal'_{-i}}\given \signal'_i = \signal'_{i, a}}} \cdot w'_{a + 1} \cdot \prob{\experiment'_i(\val_i) \geq \signal'_{i, a + 1}} \\
    & \quad + \rbr{\expect[\experiment'',\dists]{{\alloc''_i}^{\per}\rbr{\signal''_i, \signal''_{-i}}\given \signal''_i = \signal''_{i, a + 1}} - \expect[\experiment'',\dists]{{\alloc''_i}^{\per}\rbr{\signal'_i, \signal'_{-i}}\given \signal''_i = \signal'_{i, a}}} \cdot w''_{a + 1} \cdot \prob{\experiment''_i(\val_i) \geq \signal''_{i, a + 1}} \\
    & \quad  - 2\rbr{\expect[\experiment,\dists]{\alloc_i^{\per}\rbr{\signal_i, \signal_{-i}}\given \signal_i = \signal_{i, a + 1}} - \expect[\experiment,\dists]{\alloc_i^{\per}\rbr{\signal_i, \signal_{-i}}\given \signal_i = \signal_{i, a }}} \cdot w_{a + 1} \cdot \prob{\experiment_i(\val_i) \geq \signal_{i, a + 1}} \\
    & = \rbr{\expect[\experiment,\dists]{\alloc_i^{\per}\rbr{\signal_i, \signal_{-i}}\given \signal_i = \signal_{i, a + 1}} - \expect[\experiment,\dists]{\alloc_i^{\per}\rbr{\signal_i, \signal_{-i}}\given \signal_i = \signal_{i, a}}}\cdot\\
    &\qquad\frac{2\varepsilon^2 \rbr{\prob{\experiment_i(\val_i) \geq \signal_{i, a + 1}} - q_{a + 1}}\rbr{w_{a + 1} - \val_a}}{q_a^2 - \varepsilon^2}\\
    & \geq 0
\end{align*} where the second equation follows from the fact that  \(\expect[\experiment,\dists]{\alloc_i^{\per}\rbr{\signal_i, \signal_{-i}}\given \signal_i = \signal_{i, a  + 1}}\) remain unchanged and the following equation that
\[\prob{\experiment'_i(\val_i) \geq \signal'_{i, a + 1}} - \varepsilon = \prob{\experiment''_i(\val_i) \geq \signal''_{i, a + 1}} + \varepsilon = \prob{\experiment_i(\val_i) \geq \signal_{i, a + 1}}.\]

Finally, let us compute the contribution from component (3). For each signal \(\signal_{j, b}\), define 
\[\Gamma\rbr{\signal_{j, b}} := \expect{\val_i \given\experiment_j\rbr{\val_j} = \signal_{j, b}} \cdot \prob{\experiment_j(\val_j) \geq \signal_{j, b}} \]

Consequently, it follows that 
\begin{align*}
    \Delta_{(3)}' &= \varepsilon \cdot  \sum_{j = 1, j \neq i}^n \sum_{\signal_{j, b}: \signal_{i, a} <_{\per} \signal_{j, b} <_{\per} \signal_{i, a + 1}} \Gamma\rbr{\signal_{j, b}} \cdot \prod_{l\neq i,j} \rbr{1 - \sum_{\tl{\signal_{l}\in \signals_l}{\signal_l <_{\per} \signal_{j, b}} } \prob{\experiment_l\rbr{\val_l} = \signal_l}}\\
        & \quad - \varepsilon \cdot  \sum_{j = 1, j \neq i}^n \sum_{\tl{\signal_{j, b}: \signal_{i, a} <_{\per} \signal_{j, b} <_{\per} \signal_{i, a + 1}}{b < |V_j|}} \Gamma\rbr{\signal_{j, b + 1}} \cdot \prod_{l\neq i,j} \rbr{1 - \sum_{\tl{\signal_{l}\in \signals_l}{\signal_l <_{\per} \signal_{j, b}} } \prob{\experiment_l\rbr{\val_l} = \signal_l}}
\end{align*}

Similarly, it can be calculated that 
\begin{align*}
    \Delta_{(3)}'' &= - \varepsilon \cdot  \sum_{j = 1, j \neq i}^n \sum_{\signal_{j, b}: \signal_{i, a} <_{\per} \signal_{j, b} <_{\per} \signal_{i, a + 1}} \Gamma\rbr{\signal_{j, b}} \cdot \prod_{l\neq i,j} \rbr{1 - \sum_{\tl{\signal_{l}\in \signals_l}{\signal_l <_{\per} \signal_{j, b}} } \prob{\experiment_l\rbr{\val_l} = \signal_l}}\\
        & \quad + \varepsilon \cdot  \sum_{j = 1, j \neq i}^n \sum_{\tl{\signal_{j, b}: \signal_{i, a} <_{\per} \signal_{j, b} <_{\per} \signal_{i, a + 1}}{b < |V_j|}} \Gamma\rbr{\signal_{j, b + 1}} \cdot \prod_{l\neq i,j} \rbr{1 - \sum_{\tl{\signal_{l}\in \signals_l}{\signal_l <_{\per} \signal_{j, b}} } \prob{\experiment_l\rbr{\val_l} = \signal_l}}
\end{align*}

Therefore, it holds that \begin{align*}
  \Delta_{(3)}' + \Delta_{(3)}'' = 0.
\end{align*}

Notice that \(\val_a - w_a\) and \(w_{a + 1} - \val_a\) cannot both be \(0\) simultaneously, and \(\expect[\experiment,\dists]{\alloc_i^{\per}\rbr{\signal_i, \signal_{-i}}\given \signal_i = \signal_{i, a + 1}} - \expect[\experiment,\dists]{\alloc_i^{\per}\rbr{\signal_i, \signal_{-i}}\given \signal_i = \signal_{i, a}}\) is positive for all $a$. Therefore, for a sufficiently small constant \(\varepsilon>0\), it follows that 
\begin{align*}
    \Delta' + \Delta'' > 0. 
\end{align*}

This implies that deviating to a signal structure yielding strictly higher revenue is always possible. Such a possibility contradicts our initial assumption that the signal structure \(\lbrace (\signals_i, \experiment_i) \rbrace_{i\in[n]}\) is optimal. Therefore, we conclude the proof with this contradiction.
\end{proof}

\begin{proof}[Proof of \cref{lem:nph_2signal}]

{We now construct the hard instance as described in \cref{sec:nph}. For the completeness of the proof, we restate the construction in this proof.}

We provide the reduction from the following NP-complete problem \citep{ng2010product}:

\textsc{Subset Product:} Given integers \(a_1, a_2, \cdots, a_n\) with \(a_i > 1\) and a positive integer \(B\), is it possible to find a subset \(S\subseteq [n]\) such that \(\prod_{i\in S} a_i = B\)?

Consider an instance of the \textsc{Subset Product} problem, characterized by the parameters $a_1, a_2, \ldots, a_n$ and $B$. We proceed to formulate a corresponding instance of the \textsc{Optimal 2-Signal} problem with \( n \) buyers. The construction is as follows: For each buyer \( i \), we define the value distribution \( F_i \) such that:
\begin{align*}
    v_i \sim F_i, v_i = \begin{cases}
        1 & \text{w.p. }  \frac{a_i - 1}{a_i}\\
        \frac{B^2 - a_i}{B^2 + 1} & \text{w.p. } \frac{a_i - 1}{a_i^2}\\
        0 & \text{w.p. } \frac{1}{a_i^2}
    \end{cases}
\end{align*}

Let us consider the optimal 2-signal structure of the instance above. First, according to Lemma~\ref{lem:partition}, we know that the optimal signal structure must be a monotone partitional signal structure. Consequently, we know that for buyer \(i\), there are only three possible candidate signal structures: \(\left(\left\{0\right\},\left\{\frac{B^2 - a_i}{B^2 + 1}, 1\right\}\right)\),  \(\left(\left\{0, \frac{B^2 - a_i}{B^2 + 1}\right\},\left\{1\right\}\right)\), or  \(\left(\left\{0, \frac{B^2 - a_i}{B^2 + 1},1\right\}\right)\). 
Similar to the proof in \cref{lem:nph_support3}, it is clear that the there exists an optimal signal structure that does not contain \(\left(\left\{0, \frac{B^2 - a_i}{B^2 + 1},1\right\}\right)\) for any \(i\in [n]\), as simply switching from \(\left(\left\{0, \frac{B^2 - a_i}{B^2 + 1},1\right\}\right)\) to \(\left(\left\{0\right\},\left\{\frac{B^2 - a_i}{B^2 + 1}, 1\right\}\right)\) will always weakly increase the revenue. 
Let us denote by \(S \subseteq [n]\) the subset of buyers electing the signal structure \(\left(\left\{0, \frac{B^2 - a_i}{B^2 + 1}\right\}, \{1\}\right)\). Given the dichotomy of possible signal structures for each buyer, we define \(T = [n] \setminus S\) to represent the complementary subset of buyers who adopt the signal structure \(\left(\{0\}, \left\{\frac{B^2 - a_i}{B^2 + 1}, 1\right\}\right)\).

Now for any bidder in \(S\), they always receive a binary signal. For the high signal, their virtual value is simply \(1\). For the low signal, we first notice that the expected value is \(\left(\frac{B^2-a_i}{B^2+1}\cdot \frac{a_i - 1}{a_i^2}\right)/(\frac{1}{a_i}) = \frac{(a_i-1)(B^2-a_i)}{a_i(B^2+1)} \) . Consequently, their virtual value is 
\[\frac{(a_i-1)(B^2-a_i)}{a_i(B^2+1)}-\left(1 -\frac{(a_i-1)(B^2-a_i)}{a_i(B^2+1)}\cdot (a_i - 1) \right)= \frac{1-a_i^2}{B^2+1} < 0.\]

Similarly, for any bidder in \(T\), they also receives a binary signal. For the high signal, their virtual value is \(\left(\frac{a_i-1}{a_i}+\frac{B^2-a_i}{B^2+1}\frac{a_i-1}{a_i^2}\right)/(1-\frac{1}{a_i^2}) = \frac{B^2}{B^2+1}.\) For the low signal, it is straightforward that the virtual value is negative. Therefore, according to the characterization of the optimal auction by \cite{myerson1981optimal}, the optimal mechanism under such signal structure is simply a sequential posted-price mechanism: The seller first offers a price of \(1\) to the buyers in \(S\) in any order. If the item is still not sold, the seller then presents a price of \(\frac{B^2}{B^2+1}\) to the buyers in \(T\). Therefore, the optimal revenue attainable from a partition \( S, T \), which we denote by \(\mathrm{Rev}(S,T)\), can be computed as:
\begin{align*}
    \mathrm{Rev}(S,T) = 1 - \prod_{i\in S}\frac{1}{a_i} + \prod_{i\in S}\frac{1}{a_i}\left(1 - \prod_{j\in T}\frac{1}{a_j^2}\right)\cdot \frac{B^2}{B^2+1}.
\end{align*}

The revenue here is exactly the same as the reward in the proof by \cite{agrawal2020optimal}. Following the same argument, we are able to prove the NP-hardness. Define \(\gamma\) as \(\prod_{i\in T} a_i\). \cite{agrawal2020optimal} have established that the term \(\mathrm{Rev}(S,T)\) achieves its unique maximum when \(\gamma = B\). Therefore, given an optimal information structure, simply checking that whether the induced revenue is the same as the maximum at \(\gamma = B\) gives an answer to the \textsc{Subset Product} problem and this concludes the proof.
\end{proof}

%% file: appendix/algorithms.tex
\section{Missing Proofs for PTAS}
\label{apx:algorithms}


\label{apx:ptas}

\begin{proof}[Proof of \cref{lem:numberofdis}]
Following the implementation of the discretization scheme, we observe that the resultant distribution possesses support only over points that are integral multiples of $\varepsilon$ and do not surpass~$\varepsilon^{-1}$. 
Additionally, the probability mass allocated to each point must be a multiple of $\frac{1}{nm} \varepsilon^4$. Consequently, the upper bound on the number of feasible distributions is given by $\rbr{nm\varepsilon^{-4}}^{\varepsilon^{-2}}$. As for the compensatory term, it is required to be a multiple of $\frac{\epsilon}{nm}$ and is constrained to a maximum value of 1. Hence, the number of distinct possible compensatory terms is bounded by $O\rbr{nm\varepsilon^{-1}}$. Taking the product of these two factors, we thus complete our proof.
\end{proof}

\subsection{Proof of Correctness}
\label{apx:correctness}



\begin{proof}[Proof of \Cref{lem:Algorithm2CorrectNess}]
    
We first illustrate that given any partitioned signal structure \({\rbr{\signals_i, \experiment_i}}\) that induces a regular distribution \(\virtualdist_i\), there exists an implementable pairs \({\rbr{\virtualdist'_i, \compensate_i'}}\) such that \(\discretize\rbr{\virtualdist_i} = \rbr{\virtualdist'_i, \compensate_i'}\).

Let us assume that there are \(k_i\) signals in the signal structure \(\rbr{\signals_i, \experiment_i}\), which we denote as \(\signal_{i, 1}, \signal_{i, 2}, \cdots, \signal_{i, k_i}\). Define \(\rd(x, t)\) as the largest multiples of \(t\) that does not exceed \(x\). For each signal \(\signal_{i, j}\), let \(p\left(\signal_{i, j}\right)\) represent the probability of observing \(\signal_{i, j}\). Formally, this is defined as:
\[p\left(\signal_{i, j}\right) := \prob{\experiment(\val_i) = \signal_{i, j}}.\]

Recall that $\grid = \lbr{z\varepsilon \mid z \in \naturals, z\varepsilon \leq \varepsilon^{-1}}$.
Let us define the following pair of distribution and compensation term \(\virtualdist_i', \compensate_i'\) where 
\begin{align*}
\begin{split}
    \prob[\val \sim \virtualdist'_i]{\val = x} & = \begin{cases}
       \sum_{j = 1}^{k_i} \rd\rbr{p\rbr{\signal_{i,j}},\frac{\varepsilon}{n^4m}  }\cdot \indic\cbr{\virtual\rbr{\signal_{i, j}} \in \left[x, x + \varepsilon\right)} & x \in \grid \backslash\{0\} \\
        1 -  \sum_{k = 1}^{\varepsilon^{-2}}\prob[\val \sim k \varepsilon]{\val = x} & x = 0 \\
        0 & \text{otherwise}
    \end{cases}\\
    c_i' & =  \sum_{j = 1}^{k_i} \rd\rbr{p\rbr{\signal_{i,j}}\virtual\rbr{\signal_{i, j}},\frac{\varepsilon}{nm}}\cdot \indic\cbr{\virtual\rbr{\signal_{i, j}} > \varepsilon^{-1}}
\end{split}
\end{align*}

{It is clear that \(\rbr{\virtualdist_i', \compensate_i'}\) is the outcome of applying discretization scheme to \(\virtualdist_i\), i.e., \(\rbr{\virtualdist'_i, \compensate_i'} = \discretize\rbr{\virtualdist_i} \)}. According to the definition of Algorithm~\ref{alg:dps}, it follows that \(\rbr{\virtualdist_i, \compensate_i}\) is implementable, i.e., \(\rbr{\virtualdist_i, \compensate_i} \in \sets_i \). This is because that we can always make the transition of the state at the boundary of the partition and get \(\rbr{\virtualdist_i, \compensate_i}\) in \(\sets_i\). This finishes the first part of the proof.

For the opposite direction, suppose there exists an implementable pair \({\rbr{\virtualdist_i', \compensate'_i}}\) satisfying \(\rbr{\virtualdist_i', \compensate'_i} \in \sets_i\). Given that \(\rbr{\virtualdist_i', \compensate'_i}\) belongs to the implementable set  \(\sets_i\), there must exist some \(k' \leq k\) and \(j' \in [m_i]\) such that \(\dps\rbr{k',1,j,\virtualdist_i', \compensate'_i}\) evaluates to \(\mathsf{True}\). By backtracking this dynamic programming state, the property of the dynamic programming guarantees that we can always find a partitioned signal structure \(\rbr{\signals_i, \experiment_i} \) for each buyer \(i\) that ensures that the following conditions are met for all \(k\leq \varepsilon^{-2}\in \mathbb{Z}_{+}\):
\begin{itemize}
    \item The induced distribution of virtual values  \(\virtualdist_i\) is regular.
    \item \(\sum_{j = 1}^{k_i} \rd\rbr{p\rbr{\signal_{i,j}}\virtual\rbr{\signal_{i, j}},\frac{\varepsilon}{nm}}\cdot \indic\cbr{\virtual\rbr{\signal_{i, j}} > \varepsilon^{-1}} = c'_i\).
    \item \(\sum_{j = 1}^{k_i} \rd\rbr{p\rbr{\signal_{i,j}},\frac{\varepsilon}{n^4m}  }\cdot \indic\cbr{\virtual\rbr{\signal_{i, j}} \in \left[k\varepsilon, k\varepsilon + \varepsilon\right)} = \prob[\val\sim \virtualdist_i']{\val = k\varepsilon} \) for all \(k\in \mathbb{Z}_{+}\).
\end{itemize}

Furthermore, the backtracking process can efficiently run in \(\poly\left(\left(\frac{nm}{\varepsilon}\right)^{f\left(\frac{1}{\varepsilon}\right)}, m, k\right)\) time.  Again, it is clear that \(\rbr{\virtualdist_i', \compensate_i'}\) is just identical to \(\discretize\rbr{\virtualdist_i}\) by checking the definition of discretization. This concludes the proof.
\end{proof}

\subsection{Proof of Revenue Loss}\label{apx:revenue loss}

We first show that the error from decomposing a distribution $\virtualdist$ into a truncated virtual value distribution and a compensation term without discretization is at most $O(\epsilon)$. 
\begin{lemma}
\label{lem:highvals}
Given any distribution profile $\virtualdists$ such that 
$\virtualwel(\virtualdists,0) \leq 1$, 
let $(\virtualdists',\compensates')$ be the profile of truncated distribution and compensation pairs, i.e., $G'_i$ is the distribution of the random variable $v_i\cdot \indicate{v_i\leq \varepsilon^{-1}}$, where $v_i\sim G_i$, and $c_i$ is defined according to \Cref{eq:highvals0}. It holds that
\begin{align*}
\virtualwel(\virtualdists,0) \leq \virtualwel(\virtualdists',\compensates') \leq \virtualwel(\virtualdists,0) + 2\epsilon.
\end{align*}
\end{lemma}
\begin{proof}
First note that since the transformation only modifies values that are strictly larger that $\varepsilon^{-1}$, it holds that
\begin{align}\label{eq:highvals1}
    \expect[\vals\sim \virtualdists]{\max_{i\in [n]} v_i \cdot \indicate{\max_{i\in [n]} v_i \leq \varepsilon^{-1}}}
    = \expect[\vals'\sim \virtualdists']{\max_{i\in [n]} v'_i \cdot \indicate{\max_{i\in [n]} v'_i \leq \varepsilon^{-1}}}
\end{align}
We now consider the case when the maximum is strictly larger than $\varepsilon^{-1}$. It is clear that distributions $\lbr{\virtualdist_i'}_{i\in [n]}$ do not have values greater than $\varepsilon^{-1}$, i.e., 
\begin{align}\label{eq:highvals2}
    \expect[\vals'\sim \virtualdists']{\max_{i\in [n]} v'_i \cdot \indicate{\max_{i\in [n]} v'_i > \varepsilon^{-1}}} = 0.
\end{align}
Notice that for any set of non-negative numbers, its maximum is always no less than its sum. Consequently, the following inequality holds:
\begin{align}\label{eq:highvals3}
    \expect[\vals\sim \virtualdists]{\max_{i\in [n]} v_i \cdot \indicate{\max_{i\in [n]} v_i > \varepsilon^{-1}}} 
    \leq \sum_{i=1}^n \expect[\val_i\sim \virtualdist_i]{\val_i \cdot \indicate{\val_i>\varepsilon^{-1}}} = \sum_{i=1}^n \compensate_i.
\end{align}
Therefore, combining the above inequalities with the definition of $\virtualwel(\cdot)$, we have 
\begin{align*}
\virtualwel(\virtualdists,0) \leq \virtualwel(\virtualdists',\compensates').
\end{align*}

Moreover, if we only consider the case where exactly one $\val_i$ exceeds $\varepsilon^{-1}$, it holds that
\begin{align}\label{eq:highvals4}
    \begin{split}
     & \expect[\vals\sim \virtualdists]{\max_{i\in [n]} v_i \cdot \indicate{\max_{i\in [n]} v_i > \varepsilon^{-1}}} \\
     & \geq \sum_{i=1}^n\expect[\vals\sim \virtualdists]{\val_i \cdot \indicate{\val_i>\varepsilon^{-1}} \cdot \indicate{\max_{j\neq i} \val_j \leq \varepsilon^{-1}}}\\
     &= \sum_{i=1}^n\expect[\val_i\sim \virtualdist_i]{\val_i \cdot \indicate{\val_i>\varepsilon^{-1}}} \cdot \prob[\val_{-i}\sim\virtualdist_{-i}]{\max_{j\neq i} \val_j \leq \varepsilon^{-1}}\\
     & \geq \sum_{i=1}^n\expect[\val_i\sim \virtualdist_i]{\val_i \cdot \indicate{\val_i>\varepsilon^{-1}}} \cdot \prob[\vals\sim \virtualdists]{\max_{j} \val_j \leq \varepsilon^{-1}}\\
     & \geq \rbr{1 - \varepsilon} \cdot \sum_{i=1}^n\expect[\val_i\sim \virtualdist_i]{\val_i \cdot \indicate{\val_i>\varepsilon^{-1}}}
     = (1 - \varepsilon) \sum_{i=1}^n \compensate_i,
    \end{split}
\end{align} 
where the third inequality follows from Markov's inequality:
\begin{align*}
    \prob[\vals\sim \virtualdists]{\max_{j} \val_j \leq \varepsilon^{-1}} &= 1 - \prob[\vals\sim \virtualdists]{\max_{j} \val_j > \varepsilon^{-1}} \geq 1 - \varepsilon.
\end{align*}
Combining the inequalities \eqref{eq:highvals1} and \eqref{eq:highvals4}, we have 
\begin{align*}
\virtualwel(\virtualdists,0)&=\expect[\vals\sim \virtualdists]{\max_{i\in [n]} v_i \cdot \indicate{\max_{i\in [n]} v_i \leq \varepsilon^{-1}}} + \expect[\vals\sim \virtualdists]{\max_{i\in [n]} v_i \cdot \indicate{\max_{i\in [n]} v_i > \varepsilon^{-1}}}\\
& \geq \expect[\vals'\sim \virtualdists']{\max_{i\in [n]} v'_i \cdot \indicate{\max_{i\in [n]} v'_i \leq \varepsilon^{-1}}}
+ (1 - \varepsilon) \cdot \sum_{i=1}^n \compensate_i 
\geq (1 - \varepsilon) \virtualwel(\virtualdists',\compensates').
\end{align*}
Since \(\virtualwel(\virtualdists,0) \leq 1\), the inequality above implies that 
\begin{align*}
    \virtualwel(\virtualdists',\compensates') \leq (1 + 2\varepsilon)\virtualwel(\virtualdists, 0) \leq \virtualwel(\virtualdists, 0) + 2\varepsilon,
\end{align*}
and hence \cref{lem:highvals} holds.
\end{proof}

\begin{proof}[Proof of \cref{lem:discrete_algo2_loss}]

Given any distribution profile \( \virtualdists\), let $\virtualdists''$ and $\compensates''$ be the profile of distribution and compensation term pairs that is induced by $\virtualdists$ by simply decomposing the high values in $\virtualdists$ in compensation terms without discretization. \cref{lem:highvals} shows that 
\begin{align}\label{eq:correctnessproof1}
\virtualwel(\virtualdists,0) \leq \virtualwel(\virtualdists'',\compensates'')
\leq \virtualwel(\virtualdists,0) + 2\epsilon.
\end{align}

Considering a specific buyer \(i\), assume that there are \(k_i\) signals in the signal structure \(\rbr{\signals_i, \experiment_i}\), which we denote as \(\signal_{i, 1}, \signal_{i, 2}, \cdots, \signal_{i, k_i}\). Recall that \(\rd(x, t)\) is defined as the largest multiples of \(t\) that does not exceed \(x\) and \(p\left(\signal_{i, j}\right)\) represents the probability of observing \(\signal_{i, j}\). 

It is straightforward that \(\virtualdist_i\), the virtual value distribution for buyer \(i\), assigns a probability of \(p(\signal_{i, j})\) to the virtual value \(\virtual(\signal_{i, j})\) for each \(j \in [k_i]\). Given our method of redistributing higher values within \(\virtualdist_i\), it is clear that 
\begin{align*}
    \prob[\val \sim \virtualdist''_i]{\val = x} & = \begin{cases}
       \sum_{j = 1}^{k_i} p\rbr{\signal_{i,j}}\cdot \indic\cbr{\virtual\rbr{\signal_{i, j}} > \varepsilon^{-1}} & x = 0\\
        \sum_{j = 1}^{k_i} p\rbr{\signal_{i,j}}\cdot \indic\cbr{\virtual\rbr{\signal_{i, j}}  = x} & 0 < x \leq \varepsilon^{-1} \\
        0 & x < 0 \text{ or } x > \varepsilon^{-1}
    \end{cases}\\
    c_i'' & =  \sum_{j = 1}^{k_i} p\rbr{\signal_{i,j}}\virtual\rbr{\signal_{i, j}}\cdot \indic\cbr{\virtual\rbr{\signal_{i, j}} > \varepsilon^{-1}}
\end{align*}

According to the definition, \(\virtualdist_i', \compensate_i'\), i.e., the outcome of applying our discretization scheme on \(\virtualdist_i\), is defined as follows:
\begin{align}\label{eq:definition}
\begin{split}
    \prob[\val \sim \virtualdist'_i]{\val = x} & = \begin{cases}
       \sum_{j = 1}^{k_i} \rd\rbr{p\rbr{\signal_{i,j}},\frac{\varepsilon}{n^4m}  }\cdot \indic\cbr{\virtual\rbr{\signal_{i, j}} \in \left[x, x + \varepsilon\right)} & x \in \grid \backslash\{0\}\\
        1 -  \sum_{k = 1}^{\varepsilon^{-2}}\prob[\val \sim k \varepsilon]{\val = x} & x = 0 \\
        0 & \text{otherwise}
    \end{cases}\\
    c_i' & =  \sum_{j = 1}^{k_i} \rd\rbr{p\rbr{\signal_{i,j}}\virtual\rbr{\signal_{i, j}},\frac{\varepsilon}{nm}}\cdot \indic\cbr{\virtual\rbr{\signal_{i, j}} > \varepsilon^{-1}}
\end{split}
\end{align}


Since $\virtualdists'$ is stochastic dominated by $\virtualdists''$ and each term in $\compensates'$ is generated by rounding down the corresponding term in $\compensates''$, it follows that
\begin{align*}
\virtualwel(\virtualdists',\compensates') \leq \virtualwel(\virtualdists'',\compensates'')
\end{align*}

Next we provide the lower bound for $\virtualwel(\virtualdists',\compensates')$.
The definition of  the rounding down function impies that
\begin{align*}
p\rbr{\signal_{i,j}}\virtual\rbr{\signal_{i, j}} - \rd\rbr{p\rbr{\signal_{i,j}}\virtual\rbr{\signal_{i, j}},\frac{\varepsilon}{nm}} \leq \frac{\epsilon}{nm}.
\end{align*}

Upon summing the inequality above, the following naturally arises:
\begin{align}\label{eq:part1eq1}
\begin{split}
\sum_{i=1}^n \compensate_i'' - \sum_{i=1}^n  \compensate_i' & = \sum_{i = 1}^n \sum_{j = 1}^{k_i}   \rbr{p\rbr{\signal_{i,j}}\virtual\rbr{\signal_{i, j}} - \rd\rbr{p\rbr{\signal_{i,j}}\virtual\rbr{\signal_{i, j}},\frac{\varepsilon}{nm}}}\cdot \indic\cbr{\virtual\rbr{\signal_{i, j}}} \\
& \leq \frac{\epsilon}{nm} \cdot nm = \varepsilon.
\end{split}
\end{align} where the first inequality we use the fact that \(k_i\) is upper bounded by \(m\) for all \(i\in [n]\).

Next we bound the loss between $\virtualdists''$ and $\virtualdists'$ by employing a coupling between $\virtualdist_i''$ and $\virtualdist_i'$. By the definition of \(\virtualdist_i'\) and \(\virtualdist_i''\), the following is true for all $k\in \mathbb{Z}$.
\begin{align}\label{eq:part1eq10}
   \prob[\val \sim \virtualdist''_i]{\val \in \left[k\varepsilon, (k + 1)\varepsilon\right)} -\frac{\varepsilon^4}{n} \leq  \prob[\val \sim \virtualdist'_i]{\val = k\varepsilon} \leq  \prob[\val \sim \virtualdist''_i]{ \val \in \left[k\varepsilon, (k + 1)\varepsilon\right)} 
\end{align}

The coupling mechanism operates as follows: Whenever the distribution \(\virtualdist_i''\) assumes a value \(\val''_i\) within the interval \([k\varepsilon, (k + 1)\varepsilon)\) for some \(k \in \mathbb{Z}_{+}\), we correspondingly assign the rounded value \(k\varepsilon\) to \(\virtualdist_i'\). However, a residual probability mass exists, quantified as \(\prob[\val \sim \virtualdist''_i]{\val \in \left[k\varepsilon, (k + 1)\varepsilon\right)} - \prob[\val \sim \virtualdist'_i]{\val = k\varepsilon}\), which remains unaccounted for in \(\virtualdist_i''\). This excess probability mass is then allocated to the value \(0\) in \(\virtualdist_i'\). The coupling ensures that $\val''_i - \val_i' \leq {\varepsilon}$ with a probability of at least \(1 - \frac{\varepsilon^4}{n}\cdot \varepsilon^{-2} = 1 - \frac{\varepsilon^2}{n}\).

We next consider the product distributions $\virtualdist_1'' \times \cdots \times \virtualdist''_n$ and $\virtualdist_1' \times \cdots \times \virtualdist_n'$. Let $\rbr{\val''_1,\cdots,\val''_n}$ and $\rbr{\val_1',\cdots,\val_n'}$ be the realization of the product distributions. By union bound, with a probability of at least \(1 - \frac{\varepsilon^2}{n}\cdot n = 1 - {\varepsilon^2}\), it holds that $\val''_i - \val_i' \leq {\varepsilon}$ holds for all \(i\in [n].\) By taking the maximum, we have
\begin{align*}
\max_{i\in [n]} \val''_i - \max_{i\in [n]} \val'_i \leq {\varepsilon}.
\end{align*} holds with a probability of at least \(1 - \varepsilon^2\). Therefore, by taking the expectation, it follows that
\begin{align}\label{eq:part1eq2}
    \expect[\vals''\sim \virtualdists'']{ \max_{i\in [n]} \val''_i} - \expect[\vals'\sim \virtualdists']{ \max_{i\in [n]} \val'_i} \leq \varepsilon + \varepsilon^{-1} \cdot \varepsilon^{2} = 2\varepsilon.
\end{align}where we uses the fact that \(\virtualdists_i''\) and \(\virtualdists_i'\) are distributions with a support that does not exceed \(\varepsilon^{-1}\) for all \(i\in [n]\).

By synthesizing Equations~\eqref{eq:part1eq1} and~\eqref{eq:part1eq2}, we deduce that:
\begin{align*}
\virtualwel(\virtualdists', \compensates') \geq \virtualwel(\virtualdists'', \compensates'') - 2\epsilon \geq \virtualwel(\virtualdists, 0) - 3\epsilon.
\end{align*}

This, in conjunction with the established upper bound for \(\virtualwel(\virtualdists', \compensates')\), leads to the conclusion that:
\begin{align*}
    |\virtualwel(\virtualdists', \compensates') - \virtualwel(\virtualdists'', \compensates'')| \leq 3\varepsilon.
\end{align*}
Integrating this result with \Cref{eq:correctnessproof1}, we effectively conclude  our proof.
\end{proof}

\begin{proof}[Proof of \cref{lem:discrete_quant_loss}]
Since we only round down the probabilities and redistribute the remaining probabilities to value~$0$, it is clear that
\begin{align*}
     \expect[\val'\sim \max\rbr{\virtualdist_1',\virtualdist_2'}]{\val'}\leq \expect[\val\sim \max\rbr{\virtualdist_1,\virtualdist_2}]{\val}.
\end{align*}

Now we prove the lower bound. From the discretization scheme, we know that for any $k \in\naturals_+ $, it holds that 
\begin{align*}
\prob[v_1\sim \virtualdist_1]{\val_1 = k\varepsilon} - \frac{1}{nm}\varepsilon^4 
\leq \prob[v_1'\sim \virtualdist_1']{\val_1' = k\varepsilon}
\leq \prob[v_1\sim \virtualdist_1]{\val_1 = k\varepsilon},
\end{align*}
which further implies that
\begin{align*}
    \prob[v_1'\sim \virtualdist_1']{\val_1' \leq k\varepsilon} &= 1 - \sum_{k' > k} \prob[v_1'\sim \virtualdist_1']{\val_1' = k'\varepsilon} \\
    & \geq1 - \sum_{k' > k} \prob[v_1\sim \virtualdist_1]{\val_1 = k'\varepsilon} 
    = \prob[v_1\sim \virtualdist_1]{\val_1 \leq k\varepsilon}.
\end{align*} 
Moreover, the same inequalities hold for $\virtualdist_2$ and $\virtualdist_2'$.
Combining the inequalities implies that 
\begin{align*}
    &\prob[\val_1'\sim \virtualdist_1', \val_2'\sim \virtualdist_2']{\max(\val_1',\val_2') = k\varepsilon} \\
    &= \prob[\val_1'\sim \virtualdist_1']{\val_1' = k\varepsilon} \cdot \prob[\val_2'\sim \virtualdist_2']{\val_2'\leq k\varepsilon} +  \prob[\val_1'\sim \virtualdist_1']{\val_1' < k\varepsilon} \cdot\prob[\val_2'\sim \virtualdist_2']{\val_2'= k\varepsilon}\\
    & \geq \rbr{\prob[\val_1\sim \virtualdist_1]{\val_1 = k\varepsilon} - \frac{1}{nm}\varepsilon^4} \cdot \prob[\val_2\sim \virtualdist_2]{\val_2\leq k\varepsilon} +  \prob[\val_1\sim \virtualdist_1]{\val_1 < k\varepsilon} \cdot\rbr{\prob[\val_2'\sim \virtualdist_2']{v_2'= k\varepsilon} - \frac{1}{nm}\varepsilon^4}\\
    & = \prob[\val_1\sim\virtualdist_1, \val_2\sim \virtualdist_2]{\max(\val_1,\val_2) = k\varepsilon} - \frac{2}{nm}\varepsilon^4.
\end{align*}
Therefore, we can show that
\begin{align*}
\expect[\val'\sim \max\rbr{\virtualdist_1',\virtualdist_2'}]{\val'} 
&= \sum_{k = 1}^{\varepsilon^{-2}} k\varepsilon \cdot \prob[\val_1'\sim\virtualdist_1', \val_2'\sim \virtualdist_2']{\max(\val_1',\val_2') = k\varepsilon} \\
& \geq \sum_{k = 1}^{\varepsilon^{-2}} k\varepsilon  \cdot \rbr{\prob[\val_1\sim\virtualdist_1, \val_2\sim \virtualdist_2]{\max(\val_1,\val_2) = k\varepsilon} - \frac{2}{nm}\varepsilon^4}\\
& \geq \expect[\val\sim \max\rbr{\virtualdist_1,\virtualdist_2}]{\val} - \sum_{k=1}^{\varepsilon^{-2}}\varepsilon^{-1}\cdot \frac{2}{nm}\varepsilon^4
= \expect[\val\sim \max\rbr{\virtualdist_1,\virtualdist_2}]{\val} -  \frac{2\epsilon}{nm}.
\end{align*} which concludes the proof since $m\geq 1$.    
\end{proof}